\documentclass[11pt]{article}
\usepackage{fullpage}
\usepackage{amsthm, amssymb, amsmath}
\usepackage{cancel}

\usepackage{verbatim}
\usepackage{enumitem}
\usepackage{array}
\usepackage{multirow}
\usepackage{afterpage}
\usepackage{hyperref}
\usepackage{caption} 
\usepackage{setspace}
\usepackage{braket}
\usepackage{graphicx}

\usepackage[usenames,dvipsnames]{xcolor}

\newtheorem{theorem}{Theorem}[section]
\newtheorem{lemma}[theorem]{Lemma}

\newtheorem{corollary}[theorem]{Corollary}

\newtheorem{definition}[theorem]{Definition}

\def\colorful{1}

\ifnum\colorful=1

\fi
\ifnum\colorful=0

\fi

\newcommand{\R}{\mathbb{R}}
\newcommand{\E}{\mathbb{E}}
\newcommand{\1}{\mathbf{1}}
\newcommand{\A}{\mathcal{A}}
\newcommand{\Pd}{\mathcal{P}}
\newcommand{\Q}{\mathcal{Q}}
\newcommand{\X}{\mathcal{X}}
\newcommand{\C}{\mathbb{C}}
\newcommand{\eps}{\varepsilon}

\newcommand{\Rbb}{\mathbb{R}}
\newcommand{\Pc}{\mathcal{P}}

\usepackage{boxedminipage}

\newcommand{\ignore}[1]{}

\title{
  Lower Bounds on Time--Space Trade-Offs for\\Approximate Near Neighbors
}

\author{
Alexandr Andoni\\Columbia
\and
Thijs Laarhoven\\IBM Research Z\"urich
\and
Ilya Razenshteyn\\MIT CSAIL
\and 
Erik Waingarten\\Columbia
}

\begin{document}

\maketitle

\thispagestyle{empty}

\begin{abstract}
We show tight lower bounds for the entire trade-off between space and
query time for the Approximate Near Neighbor search problem. Our lower
bounds hold in a restricted model of computation, which captures all
hashing-based approaches. In particular, our lower bound matches the
upper bound recently shown in \cite{Laarhoven2015} for the random
instance on a Euclidean sphere (which we show in fact extends to the
entire space $\R^d$ using the techniques from \cite{AR-optimal}).

We also show tight, unconditional cell-probe lower bounds for
\emph{one} and \emph{two} probes, improving upon the best known bounds
from \cite{PTW10}.  In particular, this is the first space lower bound
(for any static data structure) for two probes which is \emph{not}
polynomially smaller than for one probe. To show the result for two probes, we establish and
exploit a connection to locally-decodable codes.
\end{abstract}

\hypersetup{linkcolor=magenta}
\hypersetup{linktocpage}


\onehalfspacing


\section{Introduction}

\subsection{Approximate Near Neighbor problem (ANN)}

The Near Neighbor Search problem (NNS) is a basic and fundamental problem in
computational geometry, defined as follows. We are given a dataset of
$n$ points $P$ from a metric space $(X, d_X)$ and a distance threshold
$r > 0$.  The goal is to preprocess $P$ in order to answer
\emph{near neighbor queries}: given a query point $q \in X$, return a
dataset point $p \in P$ with $d_X(q, p) \leq r$, or report that there is no
such point. The $d$-dimensional Euclidean $(\R^d, \ell_2)$ and Manhattan $(\R^d,
\ell_1)$ metric spaces have received the most attention.  Besides its
classical applications to similarity search over many types of data
(text, audio, images, etc; see~\cite{NNNIPS} for an overview), NNS has
been also recently used for cryptanalysis~\cite{L15a, l-svp-15} and
optimization~\cite{DRT11, HLM15, ZYS16}. 

The performance of a NNS data structure is often characterized by two
key metrics:
\begin{itemize}
\item the amount of memory a data structure occupies, and
\item the time it takes to answer a query.
\end{itemize}
All known time-efficient data
structures for NNS (e.g., \cite{Cl1, Mei93}) require space
exponential in the dimension $d$, which is prohibitively expensive unless
$d$ is very small.  To overcome this so-called ``curse of
dimensionality'', researchers proposed the $(c,r)$-\emph{Approximate}
Near Neighbor Search problem, or $(c,r)$-ANN. In this relaxed version,
we are given a dataset $P$
and a
distance threshold $r > 0$, as well as an approximation factor $c >
1$.  Given a query point $q$ with the promise that there is at least one
data point in $P$ within distance at most $r$ from $q$, the goal is to return a data
point $p \in P$ within a distance at most $cr$ from $q$.

This approximate version of NNS allows efficient data structures
with space and query time polynomial in $d$ and query time sublinear in $n$
\cite{KOR00, IM, I98, I-thesis, GIM, Char, DIIM,
  CR04, Pan, AC-fastJL, AI, TT-spherical, AINR-subLSH, AR-optimal,
  Pagh16-deterministic, Kap15-tradeoff, becker2016new,
  Laarhoven2015}. 
In practice, ANN algorithms are often successful for similarity search
even when one is interested in exact nearest neighbors~\cite{ADIIM,
  AILSR15}.  We refer the reader to~\cite{HIM12, AI-CACM,
  Andoni-Thesis} for a~survey of the theory of ANN, and
\cite{wssj-hsss-14,WLKC15-learningHash} for a more practical perspective.

In this paper, we study tight time--space trade-offs for ANN. Before stating
our results in Section~\ref{sec:results}, we provide more background
on the problem.

\subsection{Locality-Sensitive Hashing (LSH) and beyond}

A classic technique for ANN is \emph{Locality-Sensitive Hashing}
(LSH), introduced in~1998 by Indyk and Motwani~\cite{IM, HIM12}. The
main idea is to use \emph{random space partitions}, for which a pair of close points
(at distance at most $r$) is more likely to belong to the same part than a pair of far
points (at distance more than $cr$). Given such a partition, the data
structure splits the set $P$ according to the partition, and, given a query,
retrieves all the data points which belong to the same part as the query.
To get a high probability of success, the data structure maintains several partitions
and checks all of them during the query stage. LSH yields data structures with space $O(n^{1 + \rho} + d \cdot n)$
and query time $O(d \cdot n^{\rho})$. For a particular metric space and approximation
$c$, $\rho$ measures the quality of the random space partition.
Usually, $\rho = 1$ for $c = 1$ and $\rho \to 0$ as $c \to \infty$.

Since the introduction of LSH in \cite{IM}, subsequent research
established optimal values of the LSH exponent $\rho$ for several
metrics of interest, including $\ell_1$ and $\ell_2$. For the Hamming
distance~($\ell_1$), the optimal value is $\rho=\frac{1}{c} \pm
o(1)$~\cite{IM, MNP, OWZ11}. For the Euclidean metric~($\ell_2$), it is
$\rho=\frac{1}{c^2} \pm o(1)$~\cite{IM, DIIM, AI, MNP, OWZ11}.

More recently, it has been shown that better bounds on $\rho$ are possible
if the space partitions are \emph{allowed to depend on the
  dataset}\footnote{Let us note that the idea of data-dependent random
  space partitions is ubiquitous in practice, see,
  e.g.,~\cite{wssj-hsss-14,WLKC15-learningHash} for a survey. But the
  perspective in practice is that the given datasets are not ``worst
  case'' and hence it is possible to adapt to the additional ``nice''
  structure.}. That is, the algorithm is based on an observation that every dataset has some structure to exploit.
This more general framework of \emph{data-dependent LSH}
yields $\rho = \frac{1}{2c - 1} + o(1)$ for the $\ell_1$ distance, and
$\rho = \frac{1}{2c^2 - 1} + o(1)$ for $\ell_2$~\cite{AINR-subLSH,
  r-msthesis-14, AR-optimal}. Moreover, these bounds are known to be
tight for data-dependent LSH~\cite{ar15lower}.

\subsection{Random instances: the hardest instances}
\label{sec:rand_inst_sec}

At the core of the optimal data-dependent LSH data structure for
$\ell_1$ from~\cite{AR-optimal} is an algorithm that handles the
following random instances of ANN over Hamming space (also
known as the \emph{light bulb problem} in literature
\cite{valiant1988functionality} in the \emph{off-line} setting).
\begin{itemize}
  \item The dataset $P$ consists of $n$ independent uniformly random points from $\{-1, 1\}^d$, where $d = \omega(\log n)$;
  \item A query $q$ is generated by choosing a uniformly random data
    point $p \in P$, and flipping each coordinate of $p$ with probability $\frac{1}{2c}$ independently;
  \item The goal for a data structure is to recover the data point $p$ from the query point $q$.
\end{itemize}

At a high level, the data structure from~\cite{AR-optimal} proceeds in
two steps:
\begin{itemize}
\item it designs a (data-independent) LSH family that handles
the random instance, and 
\item it develops a reduction from a worst-case
instance to several instances that essentially look like random
instances.
\end{itemize}
Thus, random instances are the hardest for
ANN.
On the other hand, random instances have been used for the
lower bounds on ANN (more on this below), since
they must be handled by \emph{any} data structure
for $\left(c, \frac{d}{2c} + o(1)\right)$-ANN over $\ell_1$.

\subsection{Time--space trade-offs}

\label{tradeoff_sec}

LSH gives data structures with space around $n^{1 + \rho}$ and query
time around $n^{\rho}$. Since early results on LSH, the natural
question has been whether one can trade space for time and vice
versa. One can achieve \emph{polynomial} space with
\emph{poly-logarithmic} query time~\cite{IM, KOR00}, as well as
\emph{near-linear} space with \emph{sublinear} query
time~\cite{I-thesis}.  In the latter regime, \cite{Pan,
  Kap15-tradeoff} and, most recently,~\cite{Laarhoven2015} gave
subsequent improvements. We point out that the near-linear space
regime is especially relevant for practice: e.g., see
\cite{lv2007multi, AILSR15} for practical versions of the above
theoretical results.

For random instances, the best known trade-off is from \cite{Laarhoven2015}:

\begin{theorem}[Theorem 1 of \cite{Laarhoven2015}]
\label{thm:laarhovenbound}
Let $c \in (1, \infty)$. One can solve $\left(c, \frac{\sqrt{2}}{c} +
o(1)\right)$-ANN on the unit sphere $S^{d-1} \subset \Rbb^d$ equipped with
$\ell_2$ norm with query time $O(d \cdot n^{\rho_q + o(1)})$, and
space $O(n^{1 + \rho_u + o(1)} + d \cdot n)$ where
\begin{align}
\label{eqn:LaaTradeoff}
c^2 \sqrt{\rho_q} + (c^2 - 1) \sqrt{\rho_u} = \sqrt{2c^2 - 1}.
\end{align}
\end{theorem}

This data structure can handle the random Hamming instances introduced
in Section~\ref{sec:rand_inst_sec} via a standard reduction.  The
resulting time--space trade-off is:
\begin{equation}
\label{thijs_hamming}
c \sqrt{\rho_q} + (c - 1) \sqrt{\rho_u} = \sqrt{2c - 1}.
\end{equation}

For the sake of illustration, consider the setting of the Hamming distance
and approximation $c = 2$. The optimal data-dependent LSH from~\cite{AR-optimal}
gives space $n^{4/3+o(1)}$ and query time $n^{1/3 + o(1)}$. For random instances,
the above bound~(\ref{thijs_hamming})
gives the same bound as well as a smooth interpolation between the following extremes:
space $n^{1+o(1)}$ and query time $n^{3/4+o(1)}$, and space $n^{4+o(1)}$ and query time $n^{o(1)}$.

The algorithm from \cite{Laarhoven2015} can be applied to the entire
$\ell_2$ sphere (and hence, via standard reductions \`a la
\cite[Algorithm 25]{Valiant12}, to the entire space $\R^d$). However,
this direct extension degrades the quality of the
$(\rho_q, \rho_u)$ trade-off to essentially those corresponding to the
classical LSH bounds (e.g., for $\rho_q=\rho_u$, obtaining
$\rho_q=\rho_u=1/c^2+o(1)$, instead of the optimal
$\rho_q=\rho_u=1/(2c^2-1)+o(1)$). Nonetheless, it is possible to apply the worst-case--to--random-case
reduction from \cite{AR-optimal} in order to extend
Theorem~\ref{thm:laarhovenbound} to the entire $\R^d$ with the same
trade-off as~\eqref{eqn:LaaTradeoff} (see Appendices~\ref{spherical_sec}
and~\ref{apx:upper_general} for details).

Furthermore, we note that all algorithms for $\ell_2$ extend to
$\ell_p$, for $p\in (1,2)$, with $c^2$ being replaced with $c^p$ in the expressions for
the exponents $(\rho_q,\rho_u)$. This follows from the reduction shown
in \cite[Section~5.5]{Nguyen-thesis}.

\subsection{Lower bounds}

Lower bounds for NNS and ANN have also received much attention.  Such
lower bounds are almost always obtained in the \emph{cell-probe}
model~\cite{MNSW, miltersen1999cell}. In the cell-probe model one
measures the \emph{number of memory cells} the query algorithm
accesses. Despite a number of success stories, high cell-probe lower
bounds are notoriously hard to prove. In fact, there are few
techniques for proving high cell-probe lower bounds, for any (static)
data structure problem. For ANN in particular, we have no viable
techniques to prove $\omega(\log n)$ query time lower bounds. Due to
this state of affairs, one may rely on \emph{restricted} models of
computation, which nevertheless capture existing upper bounds.

Early lower bounds for NNS were shown for data structures in the {\em
  exact} or {\em deterministic} settings~\cite{BORl,CCGL, BR,
  Liu-deterNNS, jayram05match, CR04, PT, Y16a}. In~\cite{CR04, LPY16}
an almost tight cell-probe lower bound is shown for the randomized
Approximate \emph{Nearest} Neighbor Search under the~$\ell_1$
distance. In the latter problem, there is no distance threshold $r$,
and instead the goal is to find a data point that is not much further
than the \emph{closest} data point.
This
twist is the main source of hardness, and the result
is not applicable to the ANN problem as introduced above.

There are few results that show lower bounds for \emph{randomized}
data structures for the \emph{approximate} near neighbor problem (the setting
studied in the present paper).
The first such result~\cite{AIP} shows that any data structure that
solves $(1 + \eps, r)$-ANN for $\ell_1$ or $\ell_2$ using $t$ cell
probes requires space $n^{\Omega(1 / t\eps^2)}$.\footnote{The correct
  dependence on $1/\eps$ requires a stronger LSD lower bound from
  \cite{patrascu11structures}.} This result shows that the algorithms
of~\cite{IM, KOR00} are tight up to constants in the exponent for
$t=O(1)$.

In~\cite{PTW10} (following up on~\cite{PTW08}), the authors introduce
a general framework for proving lower bounds for ANN under any
metric. They show that lower bounds for ANN are implied by the
\emph{robust expansion} of the underlying metric space.
Using this framework, \cite{PTW10} show that $(c, r)$-ANN using $t$
cell probes requires space $n^{1 + \Omega(1 / tc)}$ for the Hamming
distance and $n^{1 + \Omega(1 / tc^2)}$ for the Euclidean distance
(for every $c > 1$).

Lower bounds were also shown for other metrics.  For the $\ell_\infty$
distance, \cite{ACP08} show a lower bound for deterministic ANN data
structures, matching the upper bound of \cite{I98} for decision
trees. This lower bound was later generalized to randomized data
structures \cite{PTW10, KP12-nns}. A~recent result~\cite{AV15} adapts
the framework of~\cite{PTW10} to Bregman divergences.
There are also lower bounds for restricted models: for LSH~\cite{MNP,
OWZ11, AILSR15} and for data-dependent LSH~\cite{ar15lower}.
We note that essentially all of the aforementioned lower bounds for ANN under $\ell_1$
\cite{AIP, PTW10, MNP, AILSR15, ar15lower} use
the \emph{random instance} defined in Section~\ref{sec:rand_inst_sec}
as a hard distribution.

\subsection{Our results}
\label{sec:results}

In this paper, we show both new cell-probe and restricted lower bounds
for $(c, r)$-ANN.  In all cases our lower bounds match the upper
bounds from~\cite{Laarhoven2015}. Our lower bounds use the random instance
from Section~\ref{sec:rand_inst_sec} as a hard distribution.
Via a standard reduction, we obtain
similar hardness results for $\ell_p$ with $1 < p \leq 2$ (with $c$
being replaced by $c^p$).

\subsubsection{One cell probe}

First, we show a tight (up to $n^{o(1)}$ factors) lower bound on the
space needed to solve ANN for a random instance, for query algorithms
that use a {\em single} cell probe. More formally, we
prove the following theorem:

\begin{theorem}[Section~\ref{sec:oneProbe}]
  \label{one_probe_thm}
  Any data structure that:
  \begin{itemize}
	\item solves $(c,r)$-ANN for the Hamming random instance (as defined in
          Section~\ref{sec:rand_inst_sec}) with probability $2/3$,
	\item operates on memory cells of size $n^{o(1)}$,
	\item for each query, looks up a \emph{single} cell,
  \end{itemize}
  must use at least $n^{\left(\frac{c}{c - 1}\right)^2 - o(1)}$ words of memory.
\end{theorem}

The space lower bound matches the upper bound
from~\cite{Laarhoven2015} (see also Appendix~\ref{apx:upper_general})
for $\rho_q=0$.
The previous best lower bound from~\cite{PTW10} for a single probe
was weaker by a polynomial factor.

We prove Theorem~\ref{one_probe_thm} by computing tight bounds on the
robust expansion of a hypercube $\{-1, 1\}^d$ as defined
in~\cite{PTW10}. Then, we invoke a result from~\cite{PTW10}, which
yields the desired cell probe lower bound. We obtain estimates on the
robust expansion via a combination of the~hypercontractivity
inequality and H\"{o}lder's inequality~\cite{AOBF}. Equivalently, one could obtain the same bounds by an application of the Generalized Small-Set Expansion Theorem of \cite{AOBF}.

\subsubsection{Two cell probes}

To state our results for two cell probes, we first define the {\em
  decision} version of ANN (first introduced in \cite{PTW10}).  Suppose
that with every data point $p \in P$ we associate a bit $x_p \in \{0,
1\}$. A new goal is: given a query $q \in \{-1, 1\}^d$ which is at
distance at most $r$ from a data point $p \in P$, and assuming that
$P\setminus \{p\}$ is at distance more than $cr$ from $q$, return correct $x_p$
with probability at least $2/3$.  It is easy to see that any algorithm
for $(c,r)$-ANN would solve this decision version.

We prove the following lower bound for data structures making
only two cell probes per query.

\begin{theorem}[see Section~\ref{sec:twoProbes}]
  \label{two_probe_thm}
  Any data structure that:
  \begin{itemize}
	\item solves the decision ANN for the random instance
          (Section \ref{sec:rand_inst_sec}) with probability $2/3$,
	\item operates on memory cells of size $o(\log n)$,
	\item accesses at most two cells for each query,
  \end{itemize}
  must use at least $n^{\left(\frac{c}{c - 1}\right)^2 - o(1)}$ words of memory.
\end{theorem}

Informally speaking, we show that the second cell probe cannot improve
the space bound by more than a subpolynomial factor.  To the best of our
knowledge, this is the first lower bound for the space of \emph{any}
static data structure problem without a polynomial gap between $t=1$
and $t\ge 2$ cell-probes. Previously, the highest ANN lower bound for
two queries was weaker by a polynomial factor~\cite{PTW10}. (This
remains the case even if we plug the tight bound on the robust
expansion into the
framework of~\cite{PTW10}.)  Thus, in order to obtain a higher lower
bound for $t=2$, we need to depart from the framework of \cite{PTW10}.

Our proof establishes a connection between two-query data structures
(for the decision version of ANN), and two-query locally-decodable
codes (LDC). A possibility of such a connection was suggested
in~\cite{PTW10}.  In particular, we show that a data structure
violating the lower bound from Theorem~\ref{two_probe_thm} implies an
efficient two-query LDC, which contradicts known LDC lower bounds
from~\cite{KW2004, BARW08}.

The first lower bound for unrestricted two-query LDCs was proved
in~\cite{KW2004} via a quantum argument. Later, the argument was
simplified and made \emph{classical} in~\cite{BARW08}. It turns out that
for our lower bound, we need to resort to the original quantum
argument of~\cite{KW2004} since it has a better dependence on the
noise rate a code is able to tolerate.
During the course of our proof, we do not obtain a full-fledged LDC,
but rather an object which can be called an \emph{LDC on average}. For this
reason, we are unable to use \cite{KW2004} as a black box but rather
adapt their proof to the average case.

Finally, we point out an important difference with
Theorem~\ref{one_probe_thm}: in Theorem~\ref{two_probe_thm} we allow
words to be merely of size $o(\log n)$ (as opposed to
$n^{o(1)}$). Nevertheless, for the \emph{decision version} of ANN the
upper bounds from~\cite{Laarhoven2015} hold even for such ``tiny''
words.  In fact, our techniques do not allow us to handle words of
size $\Omega(\log n)$ due to the weakness of known lower bounds for
two-query LDC for \emph{large alphabets}.  In particular, our argument
can not be pushed beyond word size $2^{\widetilde{\Theta}(\sqrt{\log
    n})}$ \emph{in principle}, since this would contradict known constructions
of two-query LDCs over large alphabets~\cite{DG15}!

\subsubsection{The general time--space trade-off}

Finally, we prove conditional lower bound on the entire time--space trade-off that is
tight (up to $n^{o(1)}$ factors), matching the upper bound from
\cite{Laarhoven2015} (see also Appendix~\ref{apx:upper_general}). Note
that---since we show polynomial query time lower bounds---proving
similar lower bounds {\em unconditionally} is far beyond the current
reach of techniques, modulo major breakthrough in cell probe lower
bounds.

Our lower bounds are proved in the following model, which can be
loosely thought of comprising all hashing-based frameworks we are
aware of:

\begin{definition}
\label{def:lip}
A {\em list-of-points data structure} for the ANN problem is defined as
follows:
\begin{itemize}
\item We fix (possibly randomly) sets $A_i\subseteq
  \{0,1\}^d$, for $i=1\ldots m$; also, with each possible query point
  $q \in \{0, 1\}^d$, we associate a (random) set of indices $I(q) \subseteq [m]$;
\item For a given dataset $P$, the data structure maintains $m$ lists of
  points $L_1, L_2, \dots, L_m$, where $L_i=P\cap A_i$;
\item On query $q$, we scan through each list $L_i$ for $i \in I(q)$
  and check whether there exists some $p \in L_i$ with $\|p - q\|_1
  \leq cr$. If it exists, return $p$.
\end{itemize}
The total space is defined as $s = m + \sum_{i=1}^m |L_i|$ and the
query time is $t = |I(q)| + \sum_{i \in I(q)} |L_i|$. 
\end{definition}

For this model, we prove the following theorem.

\begin{theorem}[see Section~\ref{sec:noCoding}]
\label{thm:noCoding}
Consider any list-of-points data structure for $(c, r)$-ANN for random
instances of $n$ points in the $d$-dimensional Hamming space with
$d=\omega(\log n)$, which achieves a total space of $n^{1 + \rho_u+o(1)}$,
and has query time $n^{\rho_q - o(1)}$, for $2/3$ success
probability. Then it must hold that:
\begin{equation}
\label{eq:power-relation}
c \sqrt{\rho_q} + (c - 1) \sqrt{\rho_u} \geq \sqrt{2c - 1}.
\end{equation}
\end{theorem}

We note that our model captures the basic hashing-based algorithms, in
particular most of the known algorithms for the high-dimensional ANN
problem \cite{KOR00, IM, I98, I-thesis, GIM, Char, DIIM, Pan,
  AC-fastJL, AI, Pagh16-deterministic, Kap15-tradeoff}, including the
recently proposed Locality-Sensitive Filters scheme from
\cite{becker2016new, Laarhoven2015}. The only data structures not
captured are the data-dependent schemes from \cite{AINR-subLSH, r-msthesis-14,
  AR-optimal}; we conjecture that the natural extension of the
list-of-point model to data-dependent setting would yield the same
lower bound.  In particular, Theorem \ref{thm:noCoding} uses the
random instance as a hard distribution, for which being
data-dependent seems to offer no advantage.  Indeed, a data-dependent
lower bound in the standard LSH regime (where $\rho_q=\rho_s$) has
been recently shown in \cite{ar15lower}, and matches~\eqref{eq:power-relation} for $\rho_s=\rho_q$.



\subsection{Other related work}

There has been a lot of recent algorithmic advances on high-dimensional
similarity search, including better algorithms for the closest pair
problem\footnote{These can be seen as the off-line version of
  NNS/ANN.}~\cite{Valiant12, JW15-closestPair,
  KKK16-fasterSubquadratic, KKKO16}, locality-sensitive
filters~\cite{becker2016new, Laarhoven2015}, LSH without false
negatives~\cite{Pagh16-deterministic, PP16}, to name just a few.



\section{Preliminaries}

We introduce a few definitions from \cite{PTW10} to setup the
nearest neighbor search problem for which we show lower bounds.

\begin{definition}
The goal of the $(c,r)$-approximate nearest neighbor problem with failure
probability $\delta$ is to construct a data structure over a set of
points $P\subset \{0,1\}^d$ supporting the following query: given any
point $q$ such that there exists some $p \in P$ with $\|q - p\|_1 \leq
r$, report some $p' \in P$ where $\|q - p'\|_1 \leq c r$ with
probability at least $1- \delta$.
\end{definition}

\begin{definition}[\cite{PTW10}]
\label{def:gns}
In the Graphical Neighbor Search problem (GNS), we are given a
bipartite graph $G = (U, V, E)$ where the dataset comes from $U$ and
the queries come from $V$. The dataset consists of pairs $P = \{ (p_i,
x_i) \mid p_i \in U, x_i \in \{0, 1\}, i \in [n] \}$. On query $q\in
V$, if there exists a unique $p_i$ with $(p_i, q) \in E$, then we want
to return $x_i$.
\end{definition}

We will sometimes use the GNS problem to prove lower bounds on
$(c,r)$-ANN as follows: we build a GNS graph $G$ by taking $U = V =
\{0,1\}^d$, and connecting two points $u\in U, v\in V$ iff they are at
a distance at most $r$ (see details in \cite{PTW10}). We will also
need to make sure that in our instances $q$ is not closer than $cr$ to
other points except the near neighbor.

\subsection{Robust Expansion}

The following is the fundamental property of a metric space that
\cite{PTW10} use to prove lower bounds.

\begin{definition}[Robust Expansion \cite{PTW10}]
For a GNS graph $G=(U,V,E)$, fix a distribution $e$ on $E\subset
U\times V$, and let $\mu$ be the marginal on $U$ and $\eta$ be the
marginal on $V$.  For $\delta, \gamma\in(0,1]$, the robust expansion
$\Phi_r(\delta, \gamma)$ is defined as follows:
$$
\Phi_r(\delta, \gamma)=\min_{A\subset V : \eta(A)\le \delta} \min_{B\subset U :
    \frac{e(A\times B)}{e(A\times V)}\ge \gamma} \frac{\mu(B)}{\eta(A)}.
$$
\end{definition}

\subsection{Locally Decodable Codes}

Finally, our 2-cell lower bounds uses results on Locally Decodable
Codes (LDCs). We present the standard definitions and results on LDCs
below, although we will need a weaker definition (and stronger
statement) for our 2-query lower bound in Section~\ref{sec:twoProbes}.

\begin{definition}
\label{def:ldc}
A $(t, \delta, \eps)$ locally decodable code (LDC) encodes $n$-bit
strings $x\in\{0,1\}^n$ into $m$-bit codewords $C(x)\in\{0,1\}^m$ such
that, for each $i\in[n]$, the bit $x_i$ can be recovered with
probability $\frac{1}{2} + \eps$ while making only $t$ queries into
$C(x)$, even if the codeword is arbitrarily modified (corrupted) in
$\delta m$ bits.
\end{definition}

We will use the following lower bound on the size of the LDCs.

\begin{theorem}[Theorem 4 from \cite{KW2004}]
If $C:\{0, 1\}^n \to \{0, 1\}^m$ is a $(2, \delta, \eps)$-LDC, then
\begin{align}
m &\geq 2^{\Omega(\delta \eps^2 n)}.
\end{align}
\end{theorem}



\section{Robust Expansion of the Hamming Space}

The goal of this section is to compute tight bounds for the robust
expansion $\Phi_r(\delta, \gamma)$ in the Hamming space of dimension
$d$, as defined in the preliminaries. We use these bounds for all of
our lower bounds in the subsequent sections.

We use the following model for generating dataset points and queries
(which is essentially the random instance from the introduction).

\begin{definition}
For any $x \in \{-1, 1\}^n$, $N_\sigma(x)$ is a probability
distribution over $\{-1, 1\}^n$ representing the neighborhood of
$x$. We sample $y \sim N_{\sigma}(x)$ by choosing $y_i \in \{-1, 1\}$ for each
coordinate $i \in [d]$. With probability $\sigma$, $y_i = x_i$. With probability
$1- \sigma$, $y_i$ is set uniformly at random.

Given any Boolean function $f:\{-1, 1\}^n \to \R$, the function $T_\sigma f:\{-1, 1\}^n \to \R$ is
\begin{align}
T_\sigma f(x) &= \mathop{\E}_{y \sim N_\sigma(x)} [f(y)]
\end{align}
\end{definition}

In the remainder of this section, will work solely on the Hamming space $V = \{-1, 1\}^d$. We let 
\[ \sigma = 1 - \frac{1}{c} \qquad  \qquad d = \omega(\log n) \] 
and $\mu$ will refer to the uniform distribution over $V$. 

The choice of $\sigma$ allows us to make the following observations. A query is generated as follows: we sample a dataset point $x$ uniformly at random and then generate the query $y$ by sampling $y \sim N_{\sigma}(x)$. From the choice of $\sigma$, $d(x, y) \leq \frac{d}{2c}(1 + o(1))$ with high probability. In addition, for every other point in the dataset $x' \neq x$, the pair $(x', y)$ is distributed as two uniformly random points (even though $y \sim N_{\sigma}(x)$, because $x$ is randomly distributed). Therefore, by taking a union-bound over all dataset points, we can conclude that with high probability, $d(x', y) \geq \frac{d}{2}(1 - o(1))$ for each $x' \neq x$. 

Given a query $y$ generated as described above, we know there exists a dataset point $x$ whose distance to the query is $d(x, y) \leq \frac{d}{2c}(1 + o(1))$. Every other dataset point lies at a distance $d(x', y) \geq \frac{d}{2}(1 - o(1))$. Therefore, the two distances are a factor of $c - o(1)$ away. 

The following lemma is the main result of this section, and we will reference this lemma in subsequent sections.

\begin{lemma}[Robust expansion]
\label{lem:robustexpansion}
In the Hamming space equipped with the Hamming norm, for any $p, q \in [1, \infty)$ where $(q-1)(p-1) = \sigma^2$, any $\gamma \in [0, 1]$ and $m \geq 1$,
\begin{align}
\Phi_r\left(\frac{1}{m}, \gamma\right) &\geq \gamma^q m^{1 + \frac{q}{p} - q}
\end{align}
\end{lemma}

The robust expansion comes from a straight forward application from small-set expansion. In fact, one can easily prove tight bounds on robust expansion via the following lemma:
\begin{theorem}[Generalized Small-Set Expansion Theorem, \cite{AOBF}]
Let $0 \leq \sigma \leq 1$. Let $A, B \subset \{-1, 1\}^n$ have volumes $\exp(-\frac{a^2}{2})$ and $\exp(-\frac{b^2}{2})$ and assume $0 \leq \sigma a \leq b \leq a$. Then
\[ \Pr_{^{\ \ \ (x, y)}_{ \sigma-\text{correlated}}}[x \in A, y \in B] \leq \exp\left( - \frac{1}{2} \frac{a^2 - 2\sigma a b + b^2}{1 - \sigma^2}\right)\]
\end{theorem}

However, we compute the robust expansion via an application of the Bonami-Beckner Inequality and H\"{o}lder's inequality. This computation gives us a bit more flexibility with respect to parameters which will become useful in subsequent sections. We now recall the necessary tools. 

\begin{theorem}[Bonami-Beckner Inequality \cite{AOBF}]
Fix $1 \leq p \leq q$ and $0 \leq \sigma \leq \sqrt{(p-1)/(q-1)}$. Any Boolean function $f:\{-1, 1\}^n \rightarrow \R$ satisfies
\begin{align}
\|T_\sigma f\|_q &\leq \|f\|_p
\end{align}
\end{theorem}

\begin{theorem}[H\"{o}lder's Inequality]
Let $f:\{-1, 1\}^n \to \R$ and $g:\{-1, 1\}^n \to \R$ be arbitrary Boolean functions. Fix $s, t \in [1, \infty)$ where $\frac{1}{s} + \frac{1}{t} = 1$. Then
\begin{align}
\langle f, g \rangle &\leq \|f\|_s \|g\|_t 
\end{align}
\end{theorem}

We will let $f$ and $g$ be indicator functions for two sets $A$ and $B$ and use a combination of the Bonami-Beckner Inequality and H\"{o}lder's Inequality to lower bound the robust expansion. The operator $T_{\sigma}$ will applied to $f$ will measure the neighborhood of set $A$. We will compute an upper bound on the correlation of the neighborhood of $A$ and $B$ (referred to as $\gamma$) with respect to the volumes of $A$ and $B$, and the expression will give a lower bound on robust expansion. 

We also need the following lemma.

\begin{lemma}
\label{lem:robustisoperimetry}
Let $p, q \in [1, \infty)$, where $(p-1)(q-1) = \sigma^2$ and $f, g: \{-1, 1\}^d \to \R$ be two Boolean functions. Then
\[ \langle T_\sigma f, g \rangle \leq \|f\|_p \|g\|_q \]
\end{lemma}

\begin{proof}
We first apply H\"{o}lder's Inequality to split the inner-product into two parts. Then we apply the Bonami-Beckner Inequality to each part.
\begin{align}
\langle T_{\sigma} f, f \rangle &= \langle T_{\sqrt{\sigma}} f, T_{\sqrt{\sigma}} g \rangle \\
	&\leq \|T_{\sqrt{\sigma}}f\|_s \|T_{\sqrt{\sigma}}g\|_t
\end{align}
We pick the parameters $s = \dfrac{p - 1}{\sigma} + 1$ and $t = \dfrac{s}{s - 1}$, so $\frac{1}{s} + \frac{1}{t} = 1$. Note that $p \leq s$ because $\sigma < 1$ and $p \geq 1$ because $(p-1)(q-1) = \sigma^2 \leq \sigma$. We have 
\begin{align} 
q \leq \dfrac{\sigma}{p-1} + 1 = t.
\end{align}
In addition, 
\begin{align}
\sqrt{\dfrac{p-1}{s - 1}} &= \sqrt{\sigma}  \qquad  &\sqrt{\dfrac{q - 1}{t - 1}} &= \sqrt{(q-1)(s-1)} \\
			  &			   &				&= \sqrt{\frac{(q-1)(p - 1)}{\sigma}} = \sqrt{\sigma}.
\end{align}
So we can apply the Bonami-Beckner Inequality to both norms. We obtain
\begin{align}
\|T_{\sqrt{\sigma}} f\|_s \| T_{\sqrt{\sigma}}g\|_t &\leq \|f\|_p \|g\|_q
\end{align}
\end{proof}

We are now ready to prove Lemma~\ref{lem:robustexpansion}.

\begin{proof}[Proof of Lemma~\ref{lem:robustexpansion}]
We use Lemma~\ref{lem:robustisoperimetry} and the definition of robust expansion. For any two sets $A, B \subset V$, let $a = \frac{1}{2^d}|A|$ and $b = \frac{1}{2^d}|B|$ be the measure of set $A$ and $B$ with respect to the uniform distribution. We refer to $\1_A:\{-1, 1\}^d \to \{0, 1\}$ and $\1_B:\{-1, 1\}^d \to \{0, 1\}$ as the indicator functions for $A$ and $B$.
\begin{align}
\gamma &= \Pr_{x\sim\mu, y \sim N_\sigma(x)} [x \in B \mid y \in A] \label{eq:min-gamma} \\
	      &= \frac{1}{a}\langle T_{\sigma}\1_A, \1_B\rangle \\
	      &\leq a^{\frac{1}{p}-1} b^{\frac{1}{q}}
\end{align}
Therefore, $\gamma^qa^{q-\frac{q}{p}} \leq b$. Let $A$ and $B$ be the minimizers of $\frac{b}{a}$ satisfying (\ref{eq:min-gamma}) and $a \leq \frac{1}{m}$.
\begin{align}
\Phi_r\left(\frac{1}{m}, \gamma\right) &= \frac{b}{a} \\
					   &\geq \gamma^qa^{q - \frac{q}{p} - 1} \\
					   &\geq \gamma^q m^{1 + \frac{q}{p} - q}.
\end{align}
\end{proof}



\section{Tight Lower Bounds for 1 Cell Probe Data Structures}
\label{sec:oneProbe}

In this section, we prove Theorem~\ref{one_probe_thm}.
Our proof relies on the main result of \cite{PTW10} for the GNS
problem:

\begin{theorem}[Theorem 1.5 \cite{PTW10}]
\label{thm:ptwbound}
There exists an absolute constant $\gamma$ such that the following
holds. Any randomized algorithm for a weakly independent instance of
GNS which is correct with probability greater than $\frac{1}{2}$ must
satisfy
\begin{align}
\dfrac{m^tw}{n} &\geq \Phi_r\left(\frac{1}{m^t}, \frac{\gamma}{t}\right)
\end{align}
\end{theorem}

\begin{proof}[Proof of Theorem~\ref{one_probe_thm}]
The bound comes from a direct application of the computation of $\Phi_r(\frac{1}{m}, \gamma)$ in Lemma~\ref{lem:robustexpansion} to the bound in Theorem~\ref{thm:ptwbound}. Setting $t = 1$ in Theorem~\ref{thm:ptwbound}, we obtain
\begin{align}
mw &\geq n \cdot \Phi_{r}\left(\frac{1}{m}, \gamma\right) \\
	&\geq n \gamma^q m^{1 + \frac{q}{p} - q}
\end{align}
for some $p, q \in [1, \infty)$ and $(p-1)(q-1) = \sigma^2$. 
Rearranging the inequality, we obtain
\begin{align}
m &\geq \dfrac{\gamma^{\frac{p}{p-1}} n^{\frac{p}{pq-q}}}{w^{\frac{p}{pq-q}}}
\end{align}
Let $p = 1 + \frac{\log \log n}{\log n}$, and $q = 1 + \sigma^2 \frac{\log n}{\log \log n}$. Then
\begin{align}
m &\geq n^{\frac{1}{\sigma^2} - o(1)}.
\end{align}
Since $\sigma = 1 - \frac{1}{c}$ and $w = n^{o(1)}$, we obtain the desired result. 
\end{proof}

\begin{corollary}
Any 1 cell probe data structures with cell size $O(\log n)$ for $c$-approximate nearest neighbors on the sphere in $\ell_2$ needs $n^{1 + \frac{2c^2 - 1}{(c^2 -1)^2}-o(1)}$ many cells. 
\end{corollary}

\begin{proof}
Each point in the Hamming space $\{-1, 1\}^d$ (after scaling by $\frac{1}{\sqrt{d}}$) can be thought of as lying on the unit sphere. If two points are a distance $r$ apart in the Hamming space, then they are $2\sqrt{r}$ apart on the sphere with $\ell_2$ norm. Therefore a data structure for a $c^2$-approximation on the sphere gives a data structure for a $c$-approximation in the Hamming space.
\end{proof}



\section{Lower Bounds for List-of-Points Data Structures}
\label{sec:noCoding}



In this section we prove Theorem~\ref{thm:noCoding}, i.e., a tight
lower bound against data structure that fall inside the
``list-of-points'' model, as defined in Def.~\ref{def:lip}.

Recall that $A_i \subset V$ is the subset of dataset points which get
placed in $L_i$. Let $B_i \subset V$ the subset of query points which
query $L_i$, this is well defined, since $B_i = \{ v \in V \mid i \in
I(v)\}$. Suppose we sample a random dataset point $u \sim V$ and then
a random query point $v$ from the neighborhood of $u$. Let
\begin{align}
\gamma_i &= \Pr[v \in B_i \mid u \in A_i]
\end{align}
and let $s_i = \mu(A_i)$. 

On instances where $n$ dataset points $\{ u_i \}_{i=1}^n$ are drawn randomly, and a query $v$ is drawn from the neighborhood of a random dataset point, we can exactly characterize the query time. 
\begin{align}
T &= \sum_{i=1}^m \1\{v \in B_i\} \left( 1 + \sum_{j=1}^n \1\{u_j \in A_i\}\right) \\
\E[T] &= \sum_{i=1}^m \mu(B_i) + \sum_{i=1}^m \gamma_i \mu(A_i) + (n-1) \sum_{i=1}^m \mu(B_i) \mu(A_i) \\
	&\geq \sum_{i=1}^m \Phi_r(s_i, \gamma_i) s_i + \sum_{i=1}^m s_i \gamma_i + (n-1) \sum_{i=1}^m \Phi_r(s_i, \gamma_i) s_i^2
\end{align}
Since the data structure succeeds with probability $\gamma$, it must be the case that
\begin{align}
\sum_{i=1}^m s_i \gamma_i &\geq \gamma = \Pr_{j \sim [n], v \sim N(u_j)}[ \exists i \in [m] : v \in B_i , u_j \in A_i] 
\end{align}
And since we use at most space $O(s)$, 
\begin{align}
n \sum_{i=1}^m s_i &\leq O(s)
\end{align}
From Lemma~\ref{lem:robustexpansion}, for any $p, q \in [1, \infty)$ where $(p-1)(q-1) = \sigma^2$ where $\sigma = 1 - \frac{1}{c}$,
\begin{align}
\E[T] &\geq \sum_{i=1}^m s_i^{q - \frac{q}{p}} \gamma_i^q + (n-1) \sum_{i=1}^m s_i^{q - \frac{q}{p} + 1} \gamma_i^q + \gamma \\
\gamma &\leq \sum_{i=1}^m s_i \gamma_i \\
O\left(\frac{s}{n}\right) &\geq \sum_{i=1}^m s_i
\end{align}
We set $S = \{i \in [m] : s_i \neq 0 \}$ and for $i \in S$, $v_i = s_i \gamma_i$. 
\begin{align}
\E[T] &\geq \sum_{i \in S} v_i^q \left(s_i^{-\frac{q}{p}} + (n-1)s_i^{-\frac{q}{p} + 1} \right) \\
	&\geq \sum_{i \in S} \left(\dfrac{\gamma}{|S|} \right)^q \left(s_i^{-\frac{q}{p}} + (n-1)s_i^{-\frac{q}{p} + 1} \right) \label{eq:intermediate}
\end{align}
where we used the fact $q \geq 1$. Consider
\begin{equation}
\label{eq:F}
F = \sum_{i \in S} \left(s_i^{-\frac{q}{p}}  + (n-1)s_i^{-\frac{q}{p} + 1} \right)
\end{equation}

We analyze three cases separately:
\begin{itemize}
\item $0 < \rho_u \leq \frac{1}{2c - 1}$
\item $\frac{1}{2c - 1} < \rho_u \leq \dfrac{2c-1}{(c-1)^2}$
\item $\rho_u = 0$.  
\end{itemize}
For the first two cases, we let 
\begin{equation}
\label{eq:p-and-q}
q = 1 - \sigma^2 + \sigma \beta \qquad p = \dfrac{\beta}{\beta - \sigma} \qquad \beta = \sqrt{\dfrac{1 - \sigma^2}{\rho_u}} 
\end{equation}
Since $0 < \rho_u \leq \dfrac{2c-1}{(c-1)^2}$, one can verify $\beta > \sigma$ and both $p$ and $q$ are at least $1$. 

\begin{lemma}
\label{lem:q-bigger-than-p}
When $\rho_u \leq \frac{1}{2c - 1}$, and $s = n^{1 + \rho_u}$,
\[ \E[T] \geq \Omega(n^{\rho_q}) \]
where $\rho_q$ and $\rho_u$ satisfy Equation~\ref{eq:power-relation}.
\end{lemma}

\begin{proof}
In this setting, $p$ and $q$ are constants, and $q \geq p$. Therefore, $\frac{q}{p} \geq 1$, so $F$ is convex in all $s_i$'s in Equation~\ref{eq:F}. So we minimize the sum by taking $s_i = O(\frac{s}{n|S|})$ and substituting in (\ref{eq:intermediate}),
\begin{align}
\E[T] &\geq \Omega\left(\dfrac{\gamma^qs^{-q/p + 1}n^{q/p}}{|S|^{q - q/p}} \right) \\
	&\geq \Omega(\gamma^q s^{1 - q} n^{q/p})
\end{align}
since $q - q/p > 0$ and $|S| \leq s$.
In addition, $p$, $q$ and $\gamma$ are constants, $\E[T] \geq \Omega(n^{\rho_q})$ where
\begin{align}
\rho_q &= (1 + \rho_u)(1 - q) + \frac{q}{p} \\
	 &= (1 + \rho_u)(\sigma^2 - \sigma \beta) + \dfrac{(1 - \sigma^2 + \sigma \beta)(\beta - \sigma)}{\beta} \\
	 &= \left( \sqrt{1 - \sigma^2} - \sqrt{\rho_u} \sigma \right)^2 \\
	 &= \left( \dfrac{\sqrt{2c-1}}{c} - \sqrt{\rho_u} \cdot \dfrac{(c-1)}{c} \right)^2
\end{align}
\end{proof}

\begin{lemma}
When $\rho_u > \frac{1}{2c - 1}$,
\[ \E[T] \geq \Omega(n^{\rho_q}) \]
where $\rho_q$ and $\rho_u$ satisfy Equation~\ref{eq:power-relation}.
\end{lemma}

\begin{proof}
We follow a similar pattern to Lemma~\ref{lem:q-bigger-than-p}. However, we may no longer assert that $F$ is convex in all $s_i$'s. 
\begin{align}
\frac{\partial F}{\partial s_i} &= \left(-\frac{q}{p}\right) s_i^{-\frac{q}{p} - 1} + \left(-\frac{q}{p} + 1 \right)(n-1) s_i^{-\frac{q}{p}}
\end{align}
The gradient is zero when each $s_i = \dfrac{q}{(p - q)(n-1)}$. Since $q < p$, this value is positive and $\sum_{i \in S} s_i \leq O\left(\frac{m}{n}\right)$ for large enough $n$. $F$ is continuous, so it is minimized exactly at that point.
So $\E[T] \geq \left(\frac{\gamma}{|S|}\right)^{q} |S| \left( \frac{q}{(p - q)(n-1)}\right)^{-\frac{q}{p}}$. Again, we maximize $|S|$ to minimize this sum since $q \geq 1$. Therefore
\begin{align}
\E[T] &\geq \left( \frac{\gamma}{s} \right)^{q} s \left(\frac{q}{(p - q)(n-1)}\right)^{-\frac{q}{p}}
\end{align}
Since $p$, $q$ and $\gamma$ are constants, $\E[T] \geq \Omega(n^{\rho_q})$ where
\[ \rho_q = (1 + \rho_u)(1 - q) + \frac{q}{p} \]
which is the same expression for $\rho_q$ as in Lemma~\ref{lem:q-bigger-than-p}.
\end{proof}

\begin{lemma}
When $\rho_u = 0$ (so $s = O(n)$), 
\[ \E[T] \geq n^{\rho_q - o(1)} \]
where $\rho_q = \dfrac{2c-1}{c^2} = 1 - \sigma^2$.
\end{lemma}

\begin{proof}
In this case, although we cannot set $p$ and $q$ as in Equation~\ref{eq:p-and-q}, 
we let
\[ q = 1 + \sigma^2 \cdot \dfrac{\log n}{\log \log n} \qquad p = 1 + \dfrac{\log \log n}{\log n}. \]
Since $q > p$, we have
\begin{align}
\E[T] &= \Omega(\gamma^q s^{1 - q}n^{\frac{q}{p}}) \\
	&= n^{1 - \sigma^2 - o(1)}
\end{align}
giving the desired expression. 
\end{proof}



\section{Tight Lower Bounds for 2 Cell Probe Data Structures}
\label{sec:twoProbes}

In this section we prove a cell probe lower bound for ANN for $t=2$ cell probes
as stated in Theorem~\ref{two_probe_thm}.

As in \cite{PTW10}, we will prove lower bounds for GNS when $U = V$ with measure $\mu$ (see Def.~\ref{def:gns}). We assume there is an underlying graph $G$
with vertex set $V$. For any particular point $p \in V$, its neighborhood $N(p)$ is the set of points with an edge to $p$ in the graph $G$.

In the 2-query GNS problem, we have a dataset $P = \{ p_i \}_{i=1}^n
\subset V$ of $n$ points as well as a bit-string $x \in \{0,
1\}^n$. We let $D$ denote a data structure with $m$ cells of $w$ bits
each. We can think of $D$ as a map $[m] \to \{0, 1\}^w$ which holds $w$ bits in each cell. $D$ will depend on the dataset $P$ as well as the bit-string $x$. The problem says that: given a query point $q \in V$, if there exists a unique  neighbor $p_i \in
N(q)$ in the dataset, we should return $x_i$ with probability at least $\frac{2}{3}$ after making two cell-probes to $D$.

\begin{theorem}
\label{thm:2-query}
There exists a constant $\gamma > 0$ such that any non-adaptive GNS data structure holding a dataset of $n \geq 1$ points which succeeds with probability $\frac{2}{3}$ using two cell probes and $m$ cells of $w$ bits satisfies
\[ \dfrac{m \log m \cdot 2^{O(w)}}{n} \geq \Omega\left(\Phi_r\left(\frac{1}{m}, \gamma\right)\right).  \]
\end{theorem}

Theorem~\ref{two_probe_thm} will follow from Theorem~\ref{thm:2-query}
together with the robust expansion bound from
Lemma~\ref{lem:robustexpansion} for the special case when probes to the data structure are non-adaptive. For the rest of this section, we prove Theorem~\ref{thm:2-query}. We will later show how to reduce adaptive algorithms losing a sub-polynomial factor in the space for $w = o(\frac{\log n}{\log \log n})$ in Section~\ref{sec:adaptivity}.

At a high-level, we will show that with a ``too-good-to-be-true" data structure with small space we can construct a weaker notion of 2-query locally-decodable code (LDC) with small noise rate using the same amount of space\footnote{A 2-query LDC corresponds to LDCs which make two probes to their memory contents. Even though there is a slight ambiguity with the data structure notion of query, we say ``2-query LDCs" in order to be consistent with the LDC literature.}. Even though we our notion of LDC is weaker than Def.~\ref{def:ldc}, we can use most of the tools for showing 2-query LDC lower bounds from \cite{KW2004}. These arguments use quantum information theory arguments, which are very robust and still work with the 2-query weak LDC we construct. 

We note that \cite{PTW10} was the first to suggest the connection between nearest neighbor search and locally-decodable codes. This work represents the first concrete connection which gives rise to better lower bounds. 


\paragraph{Proof structure.} 
The proof of Theorem~\ref{thm:2-query} proceeds in six steps.

\begin{enumerate}
\item First we will use Yao's principle to reduce to the case of deterministic non-adaptive data structures for GNS with two cell-probes. We will give distributions over $n$-point datasets $P$, as well as bit-strings $x$ and a query $q$. After defining these distributions, we will assume the existence of a deterministic data structure which makes two cell-probes non-adaptively and succeeds with probability at least $\frac{2}{3}$ when the inputs are sampled according to the three distributions.
\item We will modify the deterministic data structure in order to get ``low-contention" data structures. These are data structures which do not rely on any single cell too much similar to Def. 6.1 in \cite{PTW10}. This will be a simple argument where we increase the space bound by a constant factor to achieve this guarantee.
\item In the third step, we will take a closer look at how the low-contention data structure  probes the cells. We will use ideas from \cite{PTW10} to understand how queries neighboring particular dataset points probe various cells of the data structure. We will conclude with finding a fixed $n$-point dataset $P$. A constant fraction of the points in the dataset will satisfy the following condition: many queries in the neighborhood of these points probe disjoint pairs of cells. Intuitively, this means information about these dataset points must be spread out over various cells. 
\item We will show that for the fixed dataset $P$, we could still recover a constant fraction bits with significant probability even if we corrupt the contents of some cells. This will be the crucial connection between nearest neighbor data structures and LDCs.
\item We will reduce to the case of $1$-bit words in order to apply the LDC arguments from \cite{KW2004}. We will increase the number of cells by a factor of $2^w$ and decrease the probability of success from $\frac{1}{2} + \eta$ to $\frac{1}{2} + \frac{\eta}{2^{2w}}$. 
\item Finally, we will design an LDC with weaker guarantees and use the arguments in \cite{KW2004} to prove lower bounds on the space of the weak LDC.
\end{enumerate}

\subsection{Deterministic Data Structure}

\begin{definition}
\label{def:rand-gns-ds}
A non-adaptive randomized algorithm $R$ for the GNS problem with two cell-probes is an algorithm specified by the following three components. The data structure preprocesses a dataset $P = \{ p_i \}_{i=1}^n$ consisting of $n$ points, as well as a bit-string $x \in \{0, 1\}^n$, in order to produce a data structure $D:[m] \to \{0, 1\}^w$ which depends on $P$ and $x$. On a query $q$, $R(q)$ chooses two indices $(i, j) \in [m]^2$, and specifies a function $f_q : \{0, 1\}^w \times \{0, 1\}^w \to \{0, 1\}$. The output is given as $f_q(D_j, D_k)$. We require that
\[ \Pr_{R, D}[f_q(D_j, D_k) = x_i] \geq \frac{2}{3} \]
whenever $q \in N(p_i)$ and $p_i$ is the unique such neighbor.
\end{definition}

Note that the indices $(i, j)$ which $R$ generates to probe the data structure as well as the function $f_q$ is independent of $P$ and $x$. 

\begin{definition}
We define the following distributions:
\begin{itemize}
\item Let $\Pd$ be the distribution over $n$-point datasets given by sampling $n$ times from our space $V$ uniformly at random. 
\item Let $\X$ be the uniform distribution over $\{0, 1\}^n$.
\item Let $\Q(P)$ be the distribution over queries given by first picking a dataset point $p \in P$ uniformly at random and then picking $q \in N(p)$ uniformly at random.
\end{itemize}
\end{definition}

\begin{lemma}
\label{lem:det-alg}
Assume $R$ is a non-adaptive randomized algorithm for GNS using two cell-probes. Then there exists a non-adaptive deterministic algorithm $A$ for GNS using two cell-probes which also produces a data structure $D:[m] \to \{0, 1\}^w$ and on query $q$ chooses two indices $j, k \in [m]$ (again, independently of $P$ and $x$) to probe in $D$ as well as a function $f_q:\{0, 1\}^w \times \{0, 1\}^w \to \{0, 1\}$ where
\[ \Pr_{P \sim \Pd, x \sim \X, q \sim \Q(P)}[ f_q(D_j, D_k) = x_i ] \geq \frac{2}{3}. \] 
\end{lemma}

\begin{proof}
The following is a direct application of Yao's principle to the success probability of the algorithm. By assumption, there exists a distribution over algorithms which can achieve probability of success at least $\frac{2}{3}$ for any single query. Therefore, for the fixed distributions $\Pd, \X,$ and $\Q$, there exists a deterministic algorithm achieving at least the same success probability. 
\end{proof}

In order to simplify notation, for any algorithm $A$, we let $A^D(q)$
denote output of the algorithm. When we write $A^D(q)$, we assume that
$A(q)$ outputs a pair of indices $(j, k)$ as well as the function
$f_q: \{0, 1\}^w \times \{0, 1\}^w \to \{0, 1\}$, and the algorithm
outputs $f_q(D_j, D_k)$. For any fixed dataset $P = \{ p_i \}_{i=1}^n$
and bit-string $x \in \{0, 1\}^n$, we have 
\[ \Pr_{q \sim N(p_i)}[A^D(q) = x_i] = \Pr_{q \sim N(p_i)}[f_q(D_j, D_k) = x_i] \]
by definition. This allows us to succinctly state the probability of correctness when the query is a neighbor of $p_i$ without caring about the specific cells the algorithm probes or the function $f_q$ the algorithm uses to make its decision.

The important thing to note is that the contents of the data structure $D$ may depend on the dataset $P$ and the bit-string $x$. However, the algorithm $A$ which produces $D$ as well as the indexes for the probes to $D$ for any query point is deterministic. 

From now on, we will assume the existence of a non-adaptive
deterministic algorithm $A$ with success probability at least
$\frac{2}{3}$ using $m$ cells of width $w$. The success probability is
taken over the random choice of the dataset $P \sim \Pd$, $x \sim \X$ and $q \sim \Q(P)$. 

\subsection{Making Low-Contention Data Structures}

For any $t \in \{1, 2\}$ and $j \in [m]$, let $A_{t, j}$ be the set of queries which probe cell $j$ at the $t$-th probe of algorithm $A$. 
These sets are well defined independently of the dataset $P$ and the bit-string $x$. In particular, we could write 
\[ A_{t, j} = \{ q \in V \mid A \text{ probes cell $j$ in probe $t$ when querying $q$ }\} \]
by running the ``probing" portion of the algorithm without the need to specify a dataset $P$ or bit-string $x$. We could write down $A_{t, j}$ by simply trying every query point $q$ and seeing which cells the algorithm probes. 
 
 In other words, since the algorithm is deterministic, the probing portion of algorithm
 $A$ is completely specified by two collections $\A_1 = \{ A_{1, j}
 \}_{j\in [m]}$ and $\A_2 = \{ A_{2, j} \}_{j \in [m]}$ as well as the function $f_q$. $\A_1$ and $\A_2$ are two
 partitions of the query space $V$. On query $q$, if $q \in A_{t, j}$,
 we make the $t$-th probe to cell $j$. We output the value of $f_q$ after observing the contents of the cells.

We now define the notion of low-contention data structures, which
informally requires the data structure not rely on any one
particular cell too much, namely no $A_{t, j}$ is too large.

\begin{definition}
A deterministic non-adaptive algorithm $A$ using $m$ cells has {\em low contention} if every set $\mu(A_{t, j}) \leq \frac{1}{m}$ for $t \in \{1, 2\}$ and $j \in [m]$.   
\end{definition}

We now use the following lemma to argue that up to a small increase in
space, a data structure can be made low-contention. 

\begin{lemma}
\label{lem:heavycells}
Suppose $A$ is a deterministic non-adaptive algorithm for GNS with two cell-probes using $m$ cells, then there exists an deterministic non-adaptive algorithm $A'$ for GNS with two cell-probes using $3m$ cells which succeeds with the same probability and has low contention.
\end{lemma}

\begin{proof}
We first handle $\A_1$ and then $\A_2$. 

Suppose $\mu(A_{1, j}) \geq \frac{1}{m}$, then we partition $A_{1, j}$ into enough parts $\{A^{(j)}_{1, k} \}_k$ of size $\frac{1}{m}$. There will be at most one set with measure between $0$ and $\frac{1}{m}$. For each of part $A^{(j)}_{1, k}$ of the partition, we make a new cell $j_k$ with the same contents as cell $j$. When a query lies inside $A^{(j)}_{1, k}$ we probe the new cell $j_k$. From the data structure side, the cell contents are replicated for all additional cells. 

The number of cells in this data structure is at most $2m$, since there can be at most $m$ cells of size $\frac{1}{m}$ and for each original cell, we have only one cell with small measure. Also, keep in mind that we have not modified the sets in $\A_2$, and thus there is at most $m$ cells for which $\mu(A_{2, j}) \geq \frac{1}{m}$.  

We do the same procedure for the second collection $\A_2$. If some $\mu(A_{2, j}) \geq \frac{1}{m}$, we partition that cell into multiple cells of size exactly $\frac{1}{m}$, with one extra small cell. Again, the total number of cells will be $m$ for dividing the heavy cells in the second probe, and at most $m$ for the lighter cells in the second probe.

We have added $m$ cells in having $\mu(A_{1, j}) \leq \frac{1}{m}$ for all $j \in [m]$, and added at most $m$ cells in order to make $\mu(A_{2, j}) \leq \frac{1}{m}$ for all $j \in [m]$. Therefore, we have at most $3m$ cells. Additionally, the contents of the cells remain the same, so the algorithm succeeds with the same probability. 
\end{proof}

Given Lemma~\ref{lem:heavycells}, we will assume that $A$ is a deterministic non-adaptive algorithm for GNS with two cell-probes using $m$ cells which has low contention. The extra factor of $3$ in the number of cells will be pushed into the asymptotic notation.

\subsection{Datasets which shatter}

We fix some $\gamma > 0$ which can be thought of as a sufficiently small constant.

\begin{definition}[Weak-shattering \cite{PTW10}]
We say a partition $A_1, \dots, A_m$ of $V$ $(K, \gamma)$-weakly shatters a point $p$ if
\[ \sum_{i \in [m]} \left( \mu(A_i \cap N(p)) - \frac{1}{K}\right)^+ \leq \gamma \]
where the operator $( \cdot )^+$ takes only the non-negative part.
\end{definition}

For a fixed dataset point $p \in P$, we refer to $\gamma$ as the ``slack" in the shattering. The slack corresponds to the total measure which is leftover after we remove an arbitrary subset of $A_{t, j} \cap N(p)$ of measure at least $\frac{1}{K}$. 

\begin{lemma}[Shattering \cite{PTW10}]
\label{lem:shattering}
Let $A_1, \dots, A_k$ collection of disjoint subsets of measure at most $\frac{1}{m}$. Then 
\[ \Pr_{p \sim \mu}[\text{$p$ is $(K, \gamma)$-weakly shattered}] \geq 1 - \gamma \]
for $K = \Phi_r\left(\frac{1}{m}, \frac{\gamma^2}{4}\right) \cdot \frac{\gamma^3}{16}$.
\end{lemma}

For the remainder of the section, we let
\[ K = \Phi_{r}\left( \frac{1}{m}, \frac{\gamma^2}{4} \right) \cdot \frac{\gamma^3}{16}. \]

We are interested in the shattering of dataset points with respect to the collections $\A_1$ and $\A_2$. The dataset points which get shattered will probe many cells in the data structure. Intuitively, a bit $x_i$ corresponding to a dataset point $p_i$ which is weakly-shattered should be stored across various cells. 

So for each point $p$ which is $(K, \gamma)$ weakly-shattered we define subsets $\beta_1, \beta_2 \subset N(p)$ which hold the ``slack" of the shattering of $p$ with respect to $\A_1$ and $\A_2$.

\begin{definition}
Let $p \in V$ be a dataset point which is $(K, \gamma)$-weakly shattered by $\A_1$ and $\A_2$. Let $\beta_1, \beta_2 \subset N(p)$ be arbitrary subsets where each $j \in [m]$ satisfies
\[ \mu(A_{1, j} \cap N(p) \setminus \beta_1) \leq \frac{1}{K} \]
and
\[ \mu(A_{2, j} \cap N(p) \setminus \beta_2) \leq \frac{1}{K} \]
Since $p$ is $(K, \gamma)$-weakly shattered, we can pick $\beta_1$ and $\beta_2$ with measure at most $\gamma$ each. We will refer to $\beta(p) = \beta_1 \cup \beta_2$.
\end{definition}

For a given collection $\A$, let $S(\A, p)$ be the event that the
collection $\A$ $(K, \gamma)$-weakly shatters $p$. Note that
Lemma~\ref{lem:shattering} implies that $\Pr_{p \sim \mu}[S(\A, p)] \geq 1 - \gamma$. 

\begin{lemma}
\label{lem:gooddbshattered}
With high probability over the choice of $n$ point dataset, at most $4\gamma n$ points do not satisfy $S(\A_1, p)$ and $S(\A_2, p)$. 
\end{lemma}

\begin{proof}
This is a simple Chernoff bound. The expected number of points $p$ which do not satisfy $S(\A_1, p)$ and $S(\A_2, p)$ is at most $2 \gamma n$. Therefore, the probability that more than $4\gamma n$ points do not satisfy $S(\A_1, p)$ and $S(\A_2, p)$ is at most $\exp\left(-\frac{2\gamma n}{3}\right)$.
\end{proof}

We call a dataset \emph{good} if there are at most $4\gamma n$ dataset points which are not $(K, \gamma)$-weakly shattered by $\A_1$ and $\A_2$.

\begin{lemma}
\label{lem:good-db}
There exists a good dataset $P = \{ p_i \}_{i=1}^n$ where
\[ \Pr_{x \sim \X, q \sim \Q(P)}[A^D(q) = x_i] \geq \frac{2}{3} - o(1) \]
\end{lemma}

\begin{proof}
  This follows via a simple argument. For any fixed dataset $P = \{ p_i \}_{i=1}^n$, let
  $$
  \textbf{P} = \Pr_{x \sim \X, q \sim Q(p)}[A^D(q) = x_i]
  $$ to simplify notation. 
\begin{align}
\frac{2}{3} &\leq \mathop{\E}_{P \sim \Pd}[ \textbf{P} ] \\
		&= (1 - o(1)) \cdot \mathop{\E}_{P \sim \Pd} [ \textbf{P} \mid \text{ $P$ is good}] + o(1) \cdot \mathop{\E}_{P\sim \Pd}[\textbf{P} \mid \text{ $P$ is not good}] \\
\frac{2}{3} - o(1) &\leq (1 - o(1)) \cdot \mathop{\E}_{P \sim \Pd}[ \textbf{P} \mid \text{$P$ is good}]
\end{align}
Therefore, there exists a dataset which is not shattered by at most $4\gamma n$ and $\Pr_{x \sim \X, q \sim \Q(P)}[A^D(y) = x_i] \geq \frac{2}{3} - o(1)$. 
\end{proof}

\subsection{Corrupting some cell contents of shattered points}

In the rest of the proof, we fix the dataset $P = \{ p_i \}_{i=1}^n$ satisfying the
conditions of Lemma~\ref{lem:good-db}, i.e., such that
\[ \Pr_{x \sim \X, q \sim \Q(P)}[ A^D(q) = x_i ] \geq \frac{2}{3} - o(1). \]

We now introduce the notion of corruption of the data structure cells
$D$, which parallels the notion of noise in locally-decodable codes.
Remember that, after fixing some bit-string $x$, the algorithm $A$
produces some data structure $D:[m] \to \{0, 1\}^w$. 

\begin{definition}
We call $D' : [m]
\to \{0, 1\}^w$ a \emph{corrupted} version of $D$ at $k$ cells if
they differ on at most $k$ cells, i.e., if $|\{ i
\in [m] : D(i) \neq D'(i)\}| \leq k$. 
\end{definition}
  
In this section, we will show there exist a dataset $P$ of $n$ points and a set $S \subset [n]$ of size $\Omega(n)$ with good recovery probability, even if the algorithm has access to a corrupted version of data structure. 

\begin{definition}
For a fixed $x \in \{0,1\}^n$, let
\[ c_x(i) = \Pr_{q \sim N(p_i)}[A^D(q) = x_i]. \]
Note that from the definitions of $\Q(P)$, $\E_{x \sim \X, i \in [n]}[c_x(i)] \geq \frac{2}{3} - o(1)$. 
\end{definition}

\begin{lemma}
\label{lem:corrupted}
Fix $\eps>0$, vector $x \in \{0, 1\}^n$, and let $D : [m] \to \{0, 1\}^w$
be the data structure the algorithm produces on dataset $P$ with
bit-string $x$. Let $D'$ be a corruption of $D$ at $\eps K$ cells. For
every $i\in[n]$ where events $S(\A_1, p_i)$ and $S(\A_2, p_i)$
occur, we have
\[ \Pr_{q \sim N(p_i)}[ A^{D'}(q) = x_i] \geq c_x(i) - 2\gamma - 2\eps. \]
\end{lemma}

\begin{proof}
Note that $c_x(i)$ represents the probability mass of queries in the
neighborhood of $p_i$ for which the algorithm returns $x_i$. We want
to understand how much of that probability mass we remove when we
avoid probing the corrupted cells.

Since the dataset point $p_i$ is $(K, \gamma)$-weakly shattered by
$\A_1$ and $\A_2$, at most $2\gamma$ probability mass of $c_i(x)$ will
come from the slack of the shattering. In more detail, if $q \sim
N(p_i)$, we have probability $c_i(x)$ that the algorithm returns
$x_i$. If we query $q \sim N(p_i) \setminus \beta(p_i)$, in the worst
case, every query $q \in \beta(p_i)$ returns $x_i$; thus, after
removing $\beta(p_i)$, we have removed at most $2\gamma$ probability
mass over queries that the algorithm returns correctly.

The remaining probability mass is distributed across various cells,
where each cell has at most $\frac{1}{K}$ mass for being probing in
the first probe, and at most $\frac{1}{K}$ mass for being probe in the
second probe. Therefore, if we remove $\eps K$ cells, the first or
second probe will probe those cells with probability at most
$2\eps$. If we avoid the $\eps K$ corrupted cells, the algorithm has
the same output as it did with the uncorrupted data structure
$D$. Therefore, the probability mass which returns $x_i$ on query $q$
in the corrupted data structure $D'$ is at least $c_x(i) - 2\gamma -
2\eps$.
\end{proof}

\begin{lemma}
\label{lem:prob-c-x}
Fix $\gamma > 0$ to be a small enough constant. There exists a set $S \subset [n]$ of size $|S| = \Omega(n)$, such that whenever $i \in S$, we have that: events $S(\A_1, p_i)$ and $S(\A_2, p_i)$ occur, and
\[ \mathop{\E}_{x \sim \X}[c_x(i)] \geq \frac{1}{2} + \nu, \]
where $\nu$ can be taken to be some small constant like $\frac{1}{10}$. 
\end{lemma}

\begin{proof}
There is at most a $4\gamma$-fraction of the dataset points which are
not shattered. For simplifying the notation, let $\textbf{P} = \Pr_{i
  \in [n]}[ \E_{x \sim\X}[c_x(i)] \geq \frac{1}{2} + \nu , S(\A_1, p_i)
\wedge S(\A_2, p_i) ]$. We need to show that $\textbf{P} = \Omega(1)$, since we will set $S \subset [n]$ as 
\[ S = \left\{ i \in [n] \mid \mathop{\E}_{x \sim \X}[c_x(i)] \geq \frac{1}{2} + \nu, S(\A_1, p_i) \wedge S(\A_2, p_i) \right\}. \]
The argument is a straight-forward averaging argument.
\begin{align}
\frac{2}{3} - o(1) &\leq \mathop{\E}_{x \sim \X, i \in [n]}[c_x(i)]\\
			 &\leq 1 \cdot 4\gamma  + 1 \cdot \textbf{P} + \left(\frac{1}{2} + \nu\right) \cdot (1 - \textbf{P}) \\
\frac{1}{6} - o(1) - 4\gamma - \nu &\leq \textbf{P} \cdot \left( \frac{1}{2} - \nu \right).
\end{align}
\end{proof}

We combine Lemma~\ref{lem:corrupted} and Lemma~\ref{lem:prob-c-x} to obtain the following condition on the dataset. 

\begin{lemma}
\label{lem:prob-corrupt}
Fix small enough $\gamma > 0$ and $\eps > 0$. There exists a set $S \subset [n]$ where $|S| = \Omega(n)$, such that whenever $i \in S$, 
\[ \mathop{\E}_{x \sim \X}\left[ \Pr_{q \sim N(p_i)}[A^{D'}(q) = x_i] \right] \geq \frac{1}{2} + \eta \]
where $\eta = \nu - 2\gamma - 2\eps$ and the algorithm probes a corrupted version of the data structure $D$. 
\end{lemma}

\begin{proof}
Consider the set $S \subset [n]$ satisfying the conditions of Lemma~\ref{lem:prob-c-x}. Whenever $i \in S$, $p_i$ gets $(K, \gamma)$-weakly shattered and on average over $x$, $A$ will recover $x_i$ with probability $\frac{1}{2} + \nu$ when probing the data structure $D$ on input $q \sim N(p_i)$, i.e
\[ \mathop{\E}_{x \sim \X} \left[ \Pr_{q \sim N(p_i)} [A^D(q) = x_i] \right] \geq \frac{1}{2} + \nu. \]
Therefore, from Lemma~\ref{lem:corrupted}, if $A$ probes $D'$ which is a corruption of $D$ in any $\eps K$ cells, $A$ will recover $x_i$ with probability at least $\frac{1}{2} + \nu - 2\gamma - 2\eps$ averaged over all $x \sim \X$ where $q \sim N(p_i)$. In other words,
\[ \mathop{\E}_{x \sim \X}\left[ \Pr_{q \sim N(p_i)}[A^{D'}(q) = x_i] \right] \geq \frac{1}{2} + \nu - 2\gamma - 2\eps. \]
\end{proof}

\begin{theorem}
\label{thm:nns-ds}
There exists an algorithm $A$ and a subset $S \subseteq [n]$ of size
$S = \Omega(n)$, where $A$ makes only 2 cell probes to
$D$. Furthermore, for any corruption of $D$ at $\eps K$ cells, $A$ can
recover $x_i$ with probability at least $\frac{1}{2} + \eta$ over the
random choice of $x \sim \X$.
\end{theorem}

\begin{proof}
In order to extract $x_i$, we generate a random query $q \sim N(p_i)$ and we probe the data structure at the cells assuming the data structure is uncorrupted. From Lemma~\ref{lem:prob-corrupt}, there exists a set $S \subset [n]$ of size $\Omega(n)$ for which this algorithm recovers $x_i$ with probability at least $\frac{1}{2} + \eta$, where the probability is taken on average over all possible $x \in \{0, 1\}^n$. 
\end{proof}

We fix the algorithm $A$ and subset $S \subset [n]$
satisfying the conditions of Theorem~\ref{thm:nns-ds}. Since we fixed
the dataset $P = \{ p_i \}_{i=1}^n$ satisfying the conditions of
Lemma~\ref{lem:good-db}, we say that $x \in \{0, 1\}^n$ is an input to
algorithm $A$ in order to initialize the data structure with dataset $P = \{p_i \}_{i=1}^n$ and $x_i$ is the bit associated with
$p_i$.

\subsection{Decreasing the word size}

We now reduce to the case when the word size is $w=1$ bit.

\begin{lemma}
\label{lem:decrease-w}
There exists a deterministic non-adaptive algorithm $A'$ which on input $x \in \{0, 1\}^n$ builds a data structure $D'$ using $m 2^w$ cells of width $1$ bit. Any $i \in S$ as well as any corruption $C$ to $D'$ in at most $\eps K$ positions satisfies
\[ \mathop{\E}_{x \in \{0, 1\}^n}\left[ \Pr_{q \sim N(p_i)} [A'^{C}(q) = x_i] \right] \geq \frac{1}{2} + \frac{\eta}{2^{2w}} \]
\end{lemma}

\begin{proof}
Given algorithm $A$ which constructs the data structure $D:[m] \to \{0, 1\}^w$ on input $x \in \{0, 1\}^n$, construct the following data structure $D' : [m \cdot 2^w] \to \{0, 1\}$. For each cell $D_j \in \{0, 1\}^w$, make $2^w$ cells which contain all the parities of the $w$ bits in $D_j$. This blows up the size of the data structure by $2^w$. 

Fix $i \in S$ and $q \in N(p_i)$ if algorithm $A$ produces a function $f_q :\{0, 1\}^w \times \{0, 1\}^w \to \{0, 1\}$ which succeeds with probability at least $\frac{1}{2} + \zeta$ over $x \in \{0, 1\}^n$, then there exists a signed parity on some input bits which equals $f_q$ in at least $\frac{1}{2} + \frac{\zeta}{2^{2w}}$ inputs $x \in \{0, 1\}^n$. Let $S_j$ be the parity of the bits of cell $j$ and $S_k$ be the parity of the bits of cell $k$. Let $f_q': \{0, 1\} \times \{0, 1\} \to \{0, 1\}$ denote the parity or the negation of the parity which equals $f_q$ on $\frac{1}{2} + \frac{\zeta}{2^{2w}}$ possible input strings $x \in \{0, 1\}^n$. 

Algorithm $A'$ will evaluate $f_{q'}$ at the cell containing the parity of the $S_j$ bits in cell $j$ and the parity of $S_k$ bits in cell $k$. Let $I_{S_j}, I_{S_k} \in [m \cdot 2^w]$ be the indices of these cells. Since we can find such function for each fixed $q \in N(p_i)$, any two cell probes to $j, k \in [m]$, and any corrupted version of $D$, the algorithm $A'$ satisfies
\[ \mathop{\E}_{x \in \{0, 1\}^n} \left[ \Pr_{q \sim N(p_i)} [f_q'(C'_{I_{S_j}}, C'_{I_{S_k}}) = x_i] \right] \geq \frac{1}{2} + \frac{\eta}{2^{2w}} \]
whenever $i \in S$. 
\end{proof}

For the remainder of the section, we will prove a version of
Theorem~\ref{thm:2-query} for algorithms with $1$-bit words. Given
Lemma~\ref{lem:decrease-w}, we will modify the space to $m \cdot
2^{w}$ and the probability to $\frac{1}{2} + \frac{\eta}{2^{2w}}$ to
obtain the answer. So for the remainder of the section, assume
algorithm $A$ has $1$ bit words.

\subsection{Connecting to Locally-Decodable Codes}

To complete the proof of Theorem~\ref{thm:2-query}, it remains to prove the
following lemma.

\begin{lemma}
\label{lem:2-query-1-bit}
Let $A$ be a non-adaptive deterministic algorithm which makes $2$ cell
probes to a data structure $D$ of $m$ cells of width $1$ bit which can
handle $\eps K$ corruptions and recover $x_i$ with probability
$\frac{1}{2} + \eta$ on random input $x \in \{0, 1\}^n$ whenever $i
\in S$ for some fixed $S$ of size $\Omega(n)$. Then the following must hold
\[ \dfrac{m \log m}{n} \geq \Omega\left(\eps K\eta^2 \right).  \]
\end{lemma}

The proof of the lemma uses \cite{KW2004} and relies heavily on
notions from quantum computing, in particular quantum information
theory as applied to LDC lower bounds.

\subsubsection{Crash Course in Quantum Computing}

We introduce a few concepts from quantum computing that are necessary
in our subsequent arguments. A \emph{qubit} is a unit-length vector in
$\C^2$. We write a qubit as a linear combination of the basis states
$(^1_0) = \ket{0}$ and $(^0_1) = \ket{1}$. The qubit $\alpha =
(^{\alpha_1}_{\alpha_2})$ can be written
\[ \ket{\alpha} = \alpha_1 \ket{0} + \alpha_2\ket{1}\]
where we refer to $\alpha_1$ and $\alpha_2$ as \emph{amplitudes} and $|\alpha_1|^2 + |\alpha_2|^2 = 1$. An $m$-\emph{qubit system} is a vector in the tensor product $\C^2 \otimes \dots \otimes \C^2$ of dimension $2^{m}$. The basis states correspond to all $2^m$ bit-strings of length $m$. For $j \in [2^m]$, we write $\ket{j}$ as the basis state $\ket{j_1} \otimes \ket{j_2} \otimes \dots \otimes \ket{j_m}$ where $j = j_1j_2\dots j_m$ is the binary representation of $j$. We will write the $m$-qubit \emph{quantum state} $\ket{\phi}$ as unit-vector given by linear combination over all $2^m$ basis states. So $\ket{\phi} = \sum_{j \in [2^m]} \phi_j \ket{j}$. As a shorthand, $\bra{\phi}$ corresponds to the conjugate transpose of a quantum state.

A \emph{mixed state} $\{ p_i, \ket{\phi_i} \}$ is a probability distribution over quantum states. In this case, we the quantum system is in state $\ket{\phi_i}$ with probability $p_i$. We represent mixed states by a density matrix $\sum p_i \ket{\phi_i} \bra{\phi_i}$. 

A measurement is given by a family of positive semi-definite operators
which sum to the identity operator. Given a quantum state $\ket{\phi}$
and a measurement corresponding to the family of operators $\{M_i^{*}
M_i\}_{i}$, the measurement yields outcome $i$ with probability $\|
M_i \ket{\phi}\|^2$ and results in state $\frac{M_i
  \ket{\phi}}{\|M_i\ket{\phi}\|^2}$, where the norm $\| \cdot \|$ is
the $\ell_2$ norm. We say the measurement makes the \emph{observation}
$M_i$.

Finally, a quantum algorithm makes a query to some bit-string $y \in
\{0, 1\}^m$ by starting with the state $\ket{c}\ket{j}$ and returning
$(-1)^{c \cdot y_j} \ket{c}\ket{j}$. One can think of $c$ as the control qubit taking values $0$ or $1$; if $c = 0$, the state remains unchanged by the query, and if $c = 1$ the state receives a $(-1)^{y_j}$ in its amplitude. The queries may be made in
superposition to a state, so the state $\sum_{c \in \{0, 1\}, j \in
  [m]} \alpha_{cj} \ket{c}\ket{j}$ becomes $\sum_{c \in \{0, 1\}, j
  \in [m]} (-1)^{c \cdot y_j}\alpha_{cj} \ket{c}\ket{j}$.

\subsubsection{Weak quantum random access codes from GNS algorithms}

\begin{definition}
$C:\{0, 1\}^n \to \{0, 1\}^m$ is a $(2, \delta, \eta)$-LDC if there exists a randomized decoding algorithm making at most $2$ queries to an $m$-bit string $y$ non-adaptively, and for all $x \in \{0, 1\}^n$, $i \in [n]$, and $y \in \{0, 1\}^m$ where $d(y, C(x)) \leq \delta m$, the algorithm can recover $x_i$ from the two queries to $y$ with probability at least $\frac{1}{2} + \eta$. 
\end{definition}

In their paper, \cite{KW2004} prove the following result about 2-query LDCs. 

\begin{theorem}[Theorem 4 in \cite{KW2004}]
\label{thm:qldc}
If $C : \{0, 1\}^n \to \{0, 1\}^m$ is a $(2, \delta, \eta)$-LDC, then $m \geq 2^{\Omega(\delta \eta^2 n)}$. 
\end{theorem}

The proof of Theorem~\ref{thm:qldc} proceeds as follows.
They show how to construct a $1$-query quantum-LDC from a classical
$2$-query LDC. From a $1$-query quantum-LDC, \cite{KW2004} constructs
a quantum random access code which encodes $n$-bit strings in $O(\log
m)$ qubits. Then they apply a quantum information theory lower bound
due to Nayak \cite{N1999}:

\begin{theorem}[Theorem 2 stated in \cite{KW2004} from Nayak \cite{N1999}]
\label{thm:nayak}
For any encoding $x \to \rho_x$ of $n$-bit strings into $m$-qubit
states, such that a quantum algorithm, given query access to $\rho_x$,
can decode any fixed $x_i$ with probability at least $1/2+\eta$, it must
hold that $m \geq (1 - H(1/2+\eta)) n$.
\end{theorem}

Our proof will follow a pattern similar to the proof of Theorem
\ref{thm:qldc}. We assume the existence of a GNS algorithm $A$ which
builds a data structure $D: [m] \to \{0, 1\}$. 
We can think of $D$ as a length $m$ binary string encoding $x$; in
particular let $D_j \in \{0, 1\}$ be the $j$th bit of $D$.

Our algorithm $A$ from Theorem~\ref{thm:nns-ds} does not satisfy the
strong properties of an LDC, preventing us from applying
\ref{thm:qldc} directly. However, it does have some LDC-\emph{ish}
guarantees. In particular, we can support $\eps K$ corruptions to
$D$. In the LDC language, this means that we can tolerate a noise
rate of $\delta = \frac{\eps K}{m}$. Additionally, we
cannot necessarily recover \emph{every} coordinate $x_i$, but we can
recover $x_i$ for $i \in S$, where $|S| = \Omega(n)$. Also, our
success probability is $\frac{1}{2} + \eta$ over the random choice of
$i \in S$ and the random choice of the bit-string $x \in \{0,
1\}^n$. Our proof follows by adapting the arguments of \cite{KW2004}
to this weaker setting.

\begin{lemma}
\label{lem:quantum-alg}
Let $r = \frac{2}{\delta a^2}$ where $\delta = \dfrac{\eps K}{m}$ and
$a\le1$ is a constant. Let $D$ be the data structure from
above (i.e., satisfying the hypothesis of Lemma~\ref{lem:2-query-1-bit}). Then there exists a quantum algorithm that, starting from the
$r(\log m + 1)$-qubit state with $r$ copies of $\ket{U(x)}$, where
\[ \ket{U(x)} = \frac{1}{\sqrt{2m}}\sum_{c \in \{0, 1\}, j \in [m]} (-1)^{c \cdot D_j} \ket{c}\ket{j} \]
can recover $x_i$ for any $i \in S$ with probability $\frac{1}{2} +
\Omega(\eta)$ (over a random choice of $x$).  
\end{lemma}

Assuming Lemma~\ref{lem:quantum-alg}, we can complete the proof of
Lemma~\ref{lem:2-query-1-bit}.

\begin{proof}[Proof of Lemma~\ref{lem:2-query-1-bit}]
The proof is similar to the proof of Theorem 2 of \cite{KW2004}.
Let $\rho_x$ represent the $s$-qubit system consisting of the $r$
copies of the state $\ket{U(x)}$, where $s = r(\log m + 1)$; $\rho_x$
is an encoding of $x$.  Using Lemma~\ref{lem:quantum-alg}, we can
assume we have a quantum algorithm that, given
$\rho_x$, can recover $x_i$ for any $i \in S$ with
probability $\alpha = \frac{1}{2} + \Omega(\eta)$ over the random
choice of $x\in \{0,1\}^n$.

We will let $H(A)$ be the Von Neumann entropy of $A$, and $H(A|B)$ be
the conditional entropy and $H(A:B)$ the mutual information.

Let $XM$ be the $(n + s)$-qubit system
\[ \frac{1}{2^n} \sum_{x \in \{0, 1\}^n} \ket{x}\bra{x} \otimes \rho_x. \]
The system corresponds to the uniform superposition of all $2^n$ strings concatenated with their encoding $\rho_x$.
Let $X$ be the first subsystem corresponding to the first $n$ qubits and $M$ be the second subsystem corresponding to the $s$ qubits. We have
\begin{align} 
H(XM) &= n + \frac{1}{2^n} \sum_{x \in \{0, 1\}^n} H(\rho_x) \geq n = H(X) \\
H(M) &\leq s,
\end{align}
since $M$ has $s$ qubits. Therefore, the mutual information $H(X : M) = H(X) + H(M) - H(XM) \leq s$. Note that $H(X | M) \leq \sum_{i=1}^n H(X_i | M)$. By Fano's inequality, if $i \in S$, 
\[ H(X_i | M) \leq H(\alpha) \]
where we are using the fact that Fano's inequality works even if we can recover $x_i$ with probability $\alpha$ averaged over all $x$'s.
Additionally, if $i \notin S$, $H(X_i | M) \leq 1$. Therefore,
\begin{align}
s \geq H(X : M) &= H(X) - H(X|M) \\
  			 &\geq H(X) - \sum_{i=1}^n H(X_i | M) \\
			 &\geq n - |S| H(\alpha) - (n - |S|) \\
			 &= |S| (1 - H(\alpha)).
\end{align}
Furthermore, $1 - H(\alpha) \geq \Omega(\eta^2)$ since, and $|S| = \Omega(n)$, we have
\begin{align}
\frac{2m}{a^2\eps K} (\log m + 1) &\geq \Omega\left(n \eta^2\right) \\
\dfrac{m \log m}{n} &\geq \Omega\left(\eps K\eta^2\right).
\end{align}
\end{proof}

It remains to prove Lemma~\ref{lem:quantum-alg}, which we proceed to
do in the rest of the section.
We first show that we can simulate our GNS algorithm with a 1-query
quantum algorithm.

\begin{lemma}
\label{lem:q-simul}
Fix an $x \in \{0, 1\}^n$ and $i \in [n]$. Let $D :[m] \to \{0, 1\}$
be the data structure produced by algorithm $A$ on input $x$. Suppose
$\Pr_{q \sim N(p_i)}[A^D(q) = x_i] = \frac{1}{2} + b$ for $b>0$. Then
there exists a quantum algorithm which makes one quantum query (to
$D$) and succeeds with probability $\frac{1}{2} + \frac{4b}{7}$ to
output $x_i$.
\end{lemma}

\begin{proof}
We use the procedure in Lemma 1 of \cite{KW2004} to determine the
output algorithm $A$ on input $x$ at index $i$. The procedure
simulates two classical queries with one quantum query.
\end{proof}

All quantum algorithms which make 1-query to $D$ can be specified in
the following manner: there is a quantum state $\ket{Q_i}$, where
\[ \ket{Q_i} = \sum_{c\in\{0, 1\}, j \in [m]} \alpha_{cj}\ket{c}\ket{j} \]
which queries $D$. After querying $D$, the resulting quantum state is
$\ket{Q_i(x)}$, where
\[ \ket{Q_i(x)} = \sum_{c\in \{0, 1\}, j \in [m]} (-1)^{c \cdot D_j} \alpha_{cj}\ket{c}\ket{j}.\]
There is also a quantum measurement $\{ R, I - R \}$ such that, after the
algorithm obtains the state $\ket{Q_i(x)}$, it performs the
measurement $\{ R, I - R\}$. If the algorithm observes $R$, it outputs
$1$ and if the algorithm observes $I - R$, it outputs 0.

From Lemma~\ref{lem:q-simul}, we know there must exist a state $\ket{Q_i}$ and $\{R, I - R\}$ where if algorithm $A$ succeeds with probability $\frac{1}{2} + \eta$ on random $x \sim \{0, 1\}^n$, then the quantum algorithm succeeds with probability $\frac{1}{2} + \frac{4\eta}{7}$ on random $x \sim \{0, 1\}^n$.

In order to simplify notation, we write $p(\phi)$ as the probability of making observation $R$ from state $\ket{\phi}$. Since $R$ is a positive semi-definite matrix, $R = M^{*}M$ and so $p(\phi) = \| M\ket{\phi} \|^2$. 

In exactly the same way as \cite{KW2004}, we can remove parts of the quantum state $\ket{Q_i(x)}$ where $\alpha_{cj} > \frac{1}{\sqrt{\delta m}} = \frac{1}{\sqrt{\eps K}}$. If we let $L = \{ (c, j) \mid \alpha_{cj} \leq \frac{1}{\sqrt{\eps K}} \}$, after keeping only the amplitudes in $L$, we obtain the quantum state $\frac{1}{a}\ket{A_i(x)}$, where
\[ \ket{A_i(x)} = \sum_{(c,j) \in L} (-1)^{c \cdot D_j}\alpha_{cj}\ket{c} \ket{j} \qquad a = \sqrt{\sum_{(c, j) \in L} \alpha_{cj}^2} \]

\begin{lemma}
\label{lem:prob-gap}
Fix $i \in S$. The quantum state $\ket{A_i(x)}$ satisfies
\[ \mathop{\E}_{x \in \{0, 1\}^n}\left[ p\left(\frac{1}{a}A_i(x) \right) \mid x_i = 1\right] - \mathop{\E}_{x \in \{0, 1\}^n}\left[ p\left(\frac{1}{a}A_i(x)\right) \mid x_i = 0\right] \geq \frac{8\eta}{7a^2}. \] 
\end{lemma}

\begin{proof}
Note that since $\ket{Q_i(x)}$ and $\{ R , I - R \}$ simulate $A$ and
succeed with probability at least $\frac{1}{2} + \frac{4\eta}{7}$ on a
random $x \in \{0, 1\}^n$, we have that
\begin{align}
\frac{1}{2}\mathop{\E}_{x \in \{0, 1\}^n}\left[ p\left( Q_i(x) \right) \mid x_i = 1\right] + \frac{1}{2} \mathop{\E}_{x \in \{0, 1\}^n}\left[ 1 - p\left(Q_i(x)\right) \mid x_i = 0 \right] &\geq \frac{1}{2} + \frac{4\eta}{7},
\end{align} 
which we can simplify to say
\begin{align}
\mathop{\E}_{x \in \{0, 1\}^n}\left[ p\left( Q_i(x) \right) \mid x_i = 1\right] + \mathop{\E}_{x \in \{0, 1\}^n} \left[ p\left(Q_i(x)\right) \mid x_i = 0 \right] &\geq \frac{8\eta}{7}. 
\end{align}

Since $\ket{Q_i(x)} = \ket{A_i(x)} + \ket{B_i(x)}$ and $\ket{B_i(x)}$
contains at most $\eps K$ parts, if all probes to $D$ in $\ket{B_i(x)}$
had corrupted values, the algorithm should still succeed with the same
probability on random inputs $x$. Therefore, the following two
inequalities hold:
\begin{align}
\mathop{\E}_{x \in \{0, 1\}^n}\left[ p\left( A_i(x) + B(x) \right) \mid x_i = 1\right] + \mathop{\E}_{x \in \{0, 1\}^n} \left[ p\left(A_i(x) + B(x) \right) \mid x_i = 0 \right] &\geq \frac{8\eta}{7} \label{eq:ineq-1}\\
\mathop{\E}_{x \in \{0, 1\}^n}\left[ p\left( A_i(x) - B(x) \right) \mid x_i = 1\right] + \mathop{\E}_{x \in \{0, 1\}^n} \left[ p\left(A_i(x) - B(x)\right) \mid x_i = 0 \right] &\geq \frac{8\eta}{7} \label{eq:ineq-2}
\end{align}
Note that $p(\phi \pm \psi) = p(\phi) + p(\psi) \pm \left(  \bra{\phi} R \ket{\psi} + \bra{\psi} D \ket{\phi}\right)$ and $p(\frac{1}{c} \phi) = \frac{p(\phi)}{c^2}$. One can verify by averaging the two inequalities (\ref{eq:ineq-1}) and (\ref{eq:ineq-2}) that we get the desired expression. 
\end{proof}

\begin{lemma}
\label{lem:q-alg-ai}
Fix $i \in S$. There exists a quantum algorithm that starting from the quantum state $\frac{1}{a}\ket{A_i(x)}$, can recover the value of $x_i$ with probability $\frac{1}{2} + \frac{2\eta}{7a^2}$ over random $x \in \{0, 1\}^n$. 
\end{lemma}

\begin{proof}
The algorithm and argument are almost identical to Theorem 3 in \cite{KW2004}, we just check that it works under the weaker assumptions. Let 
\[ q_1 = \mathop{\E}_{x \in \{0, 1\}^n}\left[p\left(\frac{1}{a}A_i(x)\right) \mid x_i = 1\right] \qquad q_0 = \mathop{\E}_{x \in \{0, 1\}^n}\left[ p\left(\frac{1}{a}A_i(x)\right) \mid x_i = 0 \right]. \]
From Lemma~\ref{lem:prob-gap}, we know $q_1 - q_0 \geq \frac{8\eta}{7a^2}$. In order to simplify notation, let $b = \frac{4\eta}{7a^2}$. So we want a quantum algorithm which starting from state $\frac{1}{a} \ket{A_i(x)}$ can recover $x_i$ with probability $\frac{1}{2} + \frac{b}{2}$ on random $x \in \{0, 1\}^n$. Assume $q_1 \geq \frac{1}{2} + b$, since otherwise $q_0 \leq \frac{1}{2} - b$ and the same argument will work for $0$ and $1$ flipped. Also, assume $q_1 + q_0 \geq 1$, since otherwise simply outputting $1$ on observation $R$ and $0$ on observation $I - R$ will work. 

The algorithm works in the following way: it outputs $0$ with probability $1 - \frac{1}{q_1 + q_0}$ and otherwise makes the measurement $\{R, I - R\}$ on state $\frac{1}{a} \ket{A_i(x)}$. If the observation made is $R$, then the algorithm outputs $1$, otherwise, it outputs $0$. The probability of success over random input $x \in \{0, 1\}^n$ is 
\begin{multline}
\mathop{\E}_{x \in \{0, 1\}^n} \left[ \Pr[\text{returns correctly}] \right] \\= \frac{1}{2}  \mathop{\E}_{x \in \{0, 1\}^n} \left[ \Pr[\text{returns 1}] \mid x_i = 1 \right] + \frac{1}{2}  \mathop{\E}_{x \in \{0, 1\}^n}\left[ \Pr[\text{returns 0}] \mid x_i = 0 \right]. \label{eq:correct}
\end{multline}
When $x_i = 1$, the probability the algorithm returns correctly is $(1 - q) p\left(\frac{1}{a} A_i(x)\right)$ and when $x_i = 0$, the probability the algorithm returns correctly is $q + (1 - q)(1 - p(\frac{1}{a}A_i(x)))$. So simplifying (\ref{eq:correct}),
\begin{align}
\mathop{\E}_{x \in \{0, 1\}^n}\left[\Pr[\text{returns correctly}] \right] &= \frac{1}{2}(1 - q) q_1 + \frac{1}{2}(q + (1 - q)(1 - q_0)) \\
												&\geq \frac{1}{2} + \frac{b}{2}.
\end{align}
\end{proof}

Now we can finally complete the proof of Lemma~\ref{lem:quantum-alg}.

\begin{proof}[Proof of Lemma~\ref{lem:quantum-alg}]
Again, the proof is exactly the same as the finishing arguments of
Theorem 3 in \cite{KW2004}, and we simply check the weaker conditions
give the desired outcome.  On input $i \in [n]$ and access to $r$
copies of the state $\ket{U(x)}$, the algorithm applies the
measurement $\{ M_i^{*} M_i, I - M_{i}^* M_i \}$ where
\[ M_{i} = \sqrt{\eps K} \sum_{(c ,j) \in L} \alpha_{cj} \ket{c, j}\bra{c , j}. \]

This measurement is designed in order to yield the state $\frac{1}{a}
\ket{A_i(x)}$ on $\ket{U(x)}$ if the measurement makes the observation
$M_i^* M_i$. The fact that the amplitudes  of $\ket{A_i(x)}$ are not
too large makes $\{ M_i^*M_i, I - M_i^*M_i\}$ a valid measurement.

The probability of observing $M_i^{*}M_i$ is $\bra{U(x)} M_i^* M_i
\ket{U(x)} = \frac{\delta a^2}{2}$, where we used that $\delta =
\frac{\eps K}{m}$. So the algorithm repeatedly applies the measurement
until observing outcome $M_i^* M_i$. If it never makes the
observation, the algorithm outputs $0$ or $1$ uniformly at random. If
the algorithm does observe $M_i^* M_i$, it runs the output of the
algorithm of Lemma~\ref{lem:q-alg-ai}. The following simple
calculation (done in \cite{KW2004}) gives the desired probability of
success on random input,
\begin{align}
\mathop{\E}_{x \in \{0, 1\}^n}\left[ \Pr[\text{returns correctly}] \right] &\geq \left(1 - (1 - \delta a^2/2)^r\right) \left(\frac{1}{2} + \frac{2\eta}{7a^2}\right) + (1 - \delta a^2/2)^r \cdot \frac{1}{2} \\
											&\geq \frac{1}{2} + \frac{\eta}{7a^2}.
\end{align}
\end{proof}

\subsubsection{On adaptivity}
\label{sec:adaptivity}

We can extend our lower bounds from the non-adaptive to the adaptive setting. 

\begin{lemma}
If there exists a deterministic data structure which makes two queries
adaptively and succeeds with probability at least $\frac{1}{2} +
\eta$, there exists a deterministic data structure which makes the two
queries non-adaptively and succeeds with probability at least
$\frac{1}{2} + \frac{\eta}{2^{w}}$.
\end{lemma}

\begin{proof}
The algorithm guesses the outcome of the first cell probe and simulates
the adaptive algorithm with the guess. After knowing which two probes
to make, we probe the data structure non-adaptively. If the algorithm
guessed the contents of the first cell-probe correctly, then we output
the value of the non-adaptive algorithm. Otherwise, we output a random
value. This algorithm is non-adaptive and succeeds with probability at
least $\left(1 - \frac{1}{2^w}\right) \cdot \frac{1}{2} +
\frac{1}{2^w} \left( \frac{1}{2} + \eta \right) = \frac{1}{2} +
\frac{\eta}{2^{w}}$.
\end{proof}

Applying this theorem, from an adaptive algorithm succeeding with probability $\frac{2}{3}$, we obtain a non-adaptive algorithm which succeeds with probability $\frac{1}{2} + \Omega(2^{-w})$. This value is lower than the intended $\frac{2}{3}$, but we the reduction to a weak LDC still goes through when let $\gamma = \Theta(2^{-w})$, $\eps = \Theta(2^{-w})$. Another consequence is that $|S| = \Omega(2^{-w} n)$. 

One can easily verify that for small enough $\gamma = \Omega(2^{-w})$, 
\[ \dfrac{m \log m \cdot 2^{\Theta(w)}}{n} \geq \Omega\left( \Phi_r\left( \frac{1}{m}, \gamma\right)\right) \]
Which yields tight lower bounds (up to sub-polynomial factors) for the Hamming space when $w = o(\log n)$. 

%
%
In the case of the Hamming space, we can compute robust expansion in a similar fashion to Theorem~\ref{one_probe_thm}. In particular, for any $p, q \in [1, \infty)$ where $(p-1)(q-1) = \sigma^2$, we have
\begin{align}
\dfrac{m \log m \cdot 2^{O(w)}}{n} &\geq \Omega(\gamma^q m^{1 + q/p - q}) \\
m^{q - q/p + o(1)} &\geq n^{1 - o(1)} \gamma^q \\
m &\geq n^{\frac{1 - o(1)}{q - q/p + o(1)}} \gamma^{\frac{q}{q - q/p + o(1)}} \\
	&= n^{\frac{p}{pq - q} - o(1)} \gamma^{\frac{p}{p - 1} - o(1)}
\end{align}
Let $p = 1 + \frac{w f(n)}{\log n}$ and $q = 1 + \sigma^2 \frac{\log n}{w f(n)}$ where we require that $w f(n) = o(\log n)$ and $f(n) \rightarrow \infty$ as $n \rightarrow \infty$. 
\begin{align}
m &\geq n^{\frac{1}{\sigma^2} - o(1)} 2^{\frac{\log n}{\log \log n}} \\
	&\geq n^{\frac{1}{\sigma^2} - o(1)}
\end{align}

\section{Acknowledgments}

We would like to thank Jop Bri\"{e}t for helping us to navigate literature about LDCs. We thank Omri Weinstein for useful discussions.


{
  \small
\singlespacing
\begin{flushleft}
\bibliographystyle{alpha}
\bibliography{bibfile}
\end{flushleft}
}

\appendix

\section{Random instances for $\ell_2$}
\label{apx:random}

We first introduce the equivalent notion of the ``random instance'' (from
Section~\ref{sec:rand_inst_sec}) for $\ell_2$. This instance is what lies at the core of
 the optimal data-dependent LSH from \cite{AR-optimal}.

\begin{itemize}
\item
All points and queries lie on a unit sphere $S^{d-1} \subset \Rbb^d$.
  \item The dataset $P$ is generated by sampling $n$ unit vectors in $S^{d-1}$ independently and uniformly at random.
  \item A query $q$ is generated by first choosing a dataset
    point $p \in P$ uniformly at random, and then choosing $q$ uniformly at random from all points in $S^{d-1}$ within distance $\frac{\sqrt{2}}{c}$ from $p$.
  \item The goal of the data structure is to preprocess $P$ so
    given a query $q$ generated as above, can recover the corresponding data point $p$.
\end{itemize}

This instance must be handled by any data structure for $\left(c + o(1), \frac{\sqrt{2}}{c}\right)$-ANN over $\ell_2$.
In fact, \cite{AR-optimal} show how to reduce any
$(c,r)$-ANN instance into several (pseudo-)random instances from
above without increasing the time and space complexity by a polynomial factor. The resulting instances are {\em pseudo-}random because they are not exactly the random instance described above, but do have roughly the same distribution over
distances from $q$ to the data points. 

Following the strategy from \cite{AR-optimal}, we first analyze the random instance, and then reduce the case for general subsets of $\Rbb^d$ to pseudo-random instances. 

\section{Spherical case}
\label{spherical_sec}

We describe how to solve a random instance of ANN on a unit
sphere $S^{d-1} \subseteq \Rbb^d$, where near neighbors are planted
within distance $\frac{\sqrt{2}}{c}$ (as defined in
Appendix~\ref{apx:random}). We obtain the same time-space tradeoff as in
\cite{Laarhoven2015}, namely~\eqref{eqn:LaaTradeoff}. In
Appendix~\ref{apx:upper_general}, we extend this algorithm to the
entire space $\R^d$ using the techniques from \cite{AR-optimal}.

Below we assume that $d = \widetilde{O}(\log n)$~\cite{JL, DG03}.

\subsection{The data structure description}

The data structure is a \emph{single} rooted $T$-ary tree consisting of $K + 1$ levels. The zeroth level holds the root $r$, and each node up to the $K$-th level has $T$ children, so there are $T^K$ leaves. For every node $v$, let $\Pc_v$  be the set of nodes on the path from $v$ to the root except the root itself.
Each node $v$ except the root
holds a random Gaussian vector $z_v \sim N(0, 1)^d$ is stored . For each node~$v$, we define the subset of the dataset $P_v \subset P$:
$$
P_v = \left\{p \in P \mid \forall v' \in \Pc_v \enspace \langle z_{v'}, p \rangle \geq \eta\right\},
$$
where $\eta > 0$ is a parameter to be chosen later. For instance, $P_r = P$, since $\Pc_r = \emptyset$. Intuitively, each set $P_v$ corresponds to a subset of the dataset which lies in the intersection of sphere caps centered around $z_{v'}$ for all $v' \in \Pc_v$. Every \emph{leaf}~$v$ of the tree stores the subset $P_v$.

To process the query $q \in S^{d-1}$, we start with the root and make our way down the tree. We consider all the children of the root $v$ with $\langle z_v, q\rangle \geq \eta'$,
where $\eta' > 0$ is a parameter to be chosen later, and recurse on them.
If we end up in a leaf $v$, we try all the points from $P_v$ until we find a near neighbor. If we don't end up in a leaf, or we do not find a neighbor, we fail.

\subsection{Analysis}

First, let us analyze the probability of success. Let $q \in S^{d-1}$ and $p \in P$ be the near neighbor ($\|p - q\| \leq \frac{\sqrt{2}}{c}$).

\begin{lemma}
\label{lem:succ-prob}
  If $$T \geq \frac{100}{\mathrm{Pr}_{z \sim N(0, 1)^d}\left[\langle z, p \rangle \geq \eta\mbox{ and }\langle z, q \rangle \geq \eta'\right]},$$ then
  the probability of successfully finding $p$ on query $q$ is at least $0.9$.
\end{lemma}
\begin{proof}
We prove this by induction. Suppose the querying
algorithm is at node $v$, where $p \in P_v$. We would like to prove that---if the conditions of
the lemma are met---the probability of success is $0.9$.
  
When $v$ is a leaf, the statement is obvious. Suppose it is true for all the children of a node $v$, then
  \begin{multline*}
  \mathrm{Pr}[\mbox{failure}] \leq \prod_{\mbox{$v'$ child of $v$}} \left(1 - \mathrm{Pr}_{z_{v'}}\left[\langle z_{v'}, p \rangle \geq \eta\mbox{ and }\langle z_{v'}, q \rangle \geq \eta'\right] \cdot 0.9\right) \\ = \left(1 - \mathrm{Pr}_{z \sim N(0, 1)^d}\left[\langle z, p \rangle \geq \eta\mbox{ and }\langle z, q \rangle \geq \eta'\right] \cdot 0.9\right)^T \leq 0.1.
  \end{multline*}
\end{proof}

Now let us understand how much space the data structure occupies. In
the lemma below, $u \in S^{d-1}$ is an arbitrary point.

\begin{lemma}
\label{lem:space}
  The expected space consumption of the data structure is at most
  $$
  n^{o(1)} \cdot T^K \left(1 + n \cdot \mathrm{Pr}_{z \sim N(0, 1)^d}\left[\langle z, u \rangle \geq \eta\right]^K\right).
  $$
\end{lemma}
\begin{proof}
  The total space the tree nodes occupy is
  $n^{o(1)} \cdot T^K$.
  
  At the same time, every point $u$ participates on average in $T^K
  \cdot \mathrm{Pr}_{z \sim N(0, 1)^d}\left[\langle z, u \rangle \geq
    \eta\right]^K$ leaves, hence the desired bound. 
\end{proof}

Finally, let us analyze the expected query time. As before, $u \in S^{d-1}$ is an arbitrary point.

\begin{lemma}
\label{lem:time}
  The expected query time is at most
  $$
  n^{o(1)} \cdot T^{K + 1} \cdot \mathrm{Pr}_{z \sim N(0, 1)^d}\left[\langle z, u \rangle \geq \eta'\right]^K \cdot \left(1 + n \cdot \mathrm{Pr}_{z \sim N(0, 1)^d}\left[\langle z, u \rangle \geq \eta\right]^K\right).
  $$
  
\end{lemma}
\begin{proof}
  First, a query touches at most $T^K \cdot \mathrm{Pr}_{z \sim N(0, 1)^d}\left[\langle z, u \rangle \geq \eta'\right]^K$ tree nodes on average .
  
  If a node is not a leaf, the time spent on it is at most $n^{o(1)} \cdot T$.
  
  For a fixed leaf and a fixed dataset point, the probability that they end up in the leaf together with the query point is
  $$
  \mathrm{Pr}_{z \sim N(0, 1)^d}\left[\langle z, u \rangle \geq \eta'\right]^K \cdot \mathrm{Pr}_{z \sim N(0, 1)^d}\left[\langle z, u \rangle \geq \eta\right]^K,
  $$
  hence we obtain the desired bound. 
\end{proof}

\subsection{Setting parameters}

First, we set $K = \sqrt{\log n}$. Second, we set $\eta > 0$ such that
$$
\mathrm{Pr}_{z \sim N(0, 1)^d}\left[\langle z, u \rangle \geq \eta\right] = n^{-1/K} = 2^{-\sqrt{\log n}}.
$$

We can simply substitute the parameter setting of Lemma~\ref{lem:space} and Lemma~\ref{lem:time} This gives an expected space of 
$$n^{o(1)} \cdot T^K$$,

and an expected query time 
$$ n^{o(1)} \cdot T^{K + 1} \cdot \mathrm{Pr}_{z \sim N(0, 1)^d}\left[\langle z, u \rangle \geq \eta'\right]^K.$$
As discussed above in Lemma~\ref{lem:succ-prob}, by setting 
$$T \geq \frac{100}{\mathrm{Pr}_{z \sim N(0, 1)^d}\left[\langle z, p \rangle \geq \eta\mbox{ and }\langle z, q \rangle \geq \eta'\right]},$$ the probability of success is $0.9$. 

In order to get the desired tradeoff, we can vary $T$ and $\eta'$. Suppose we want space to be $n^{\rho_s + o(1)}$ for $\rho_s \geq 1$. Then
we let
$$
  T = n^{\frac{\rho_s + o(1)}{K}} = 2^{\left(1 + o(1)\right) \cdot \rho_s \sqrt{\log n}},
$$
and $\eta' > 0$ to be the largest number such that for every
$p, q \in S^{d-1}$ with $\|p - q\| \leq \frac{\sqrt{2}}{c}$, we have
$$
\mathrm{Pr}_{z \sim N(0, 1)^d}\left[\langle z, p \rangle \geq \eta\mbox{ and }\langle z, q \rangle \geq \eta'\right] \geq \frac{100}{T} = 2^{-\left(1 + o(1)\right) \cdot \rho_s \sqrt{\log n}}.
$$

Again, substituting in values of Lemma~\ref{lem:time}, the query time is
$$
n^{o(1)} \cdot T^{K + 1} \cdot \mathrm{Pr}_{z \sim N(0, 1)^d}\left[\langle z, u \rangle \geq \eta'\right]^K = n^{\rho_s + o(1)} \cdot \mathrm{Pr}_{z \sim N(0, 1)^d}\left[\langle z, u \rangle \geq \eta'\right]^{\sqrt{\log n}}. 
$$

The trade-off between $\rho_s$ and $\rho_q$ follows from a standard computation of
$$
\mathrm{Pr}_{z \sim N(0, 1)^d}\left[\langle z, u \rangle \geq \eta'\right]
$$
given that
$$
\mathrm{Pr}_{z \sim N(0, 1)^d}\left[\langle z, u \rangle \geq \eta\right] = 2^{-\sqrt{\log n}}
$$
and
$$
\mathrm{Pr}_{z \sim N(0, 1)^d}\left[\langle z, p \rangle \geq \eta\mbox{ and }\langle z, q \rangle \geq \eta'\right] \geq 2^{-\left(1 + o(1)\right) \cdot \rho_s \sqrt{\log n}}.
$$

The computation is relatively standard: see~\cite{AILSR15}. We verify
next that the resulting trade-off is the same as~(\ref{eqn:LaaTradeoff}) obtained in~\cite{Laarhoven2015}.

Denote $\alpha, \beta$ to be real numbers such that $\|(1, 0) - (\alpha, \beta)\|_2 = \frac{\sqrt{2}}{c}$ and $\|(\alpha, \beta)\|_2 = 1$.
Namely, $\alpha = 1 - \frac{1}{c^2}$ and $\beta = \sqrt{1 - \alpha^2}$.

\begin{lemma}
  Suppose that $\eta, \eta' > 0$ are such that $\eta, \eta' \to \infty$ and $\frac{\eta^2 + \eta'^2 - 2 \alpha \eta \eta'}{\beta^2} \to \infty$. Then, for every $p, q \in S^{d-1}$ with $\|p - q\|_2 \leq \frac{\sqrt{2}}{c}$ one has:
  $$
  \mathrm{Pr}_{z \sim N(0, 1)^d}\left[\langle z, p \rangle \geq \eta\mbox{ and }\langle z, q \rangle \geq \eta'\right] = e^{-(1 + o(1)) \cdot \frac{\eta^2 + \eta'^2 - 2 \alpha \eta \eta'}{2 \beta^2}},
  $$
and,
  $$
  \mathrm{Pr}_{z \sim N(0, 1)^d}\left[\langle z, p \rangle \geq \eta\right] = e^{-(1 + o(1)) \cdot \frac{\eta^2}{2}}.
  $$
\end{lemma}
\begin{proof}
  Using spherical symmetry of Gaussians, we can reduce the computation to computing the Gaussian measure of the following \emph{two-dimensional} set:
  $$
  \{(x, y) \mid x \geq \eta'\mbox{ and } \alpha x + \beta y \geq \eta\}.
  $$
  The squared distance from zero to the set is:
  $$
  \frac{\eta^2 + \eta'^2 - 2 \alpha \eta \eta'}{\beta^2}.
  $$
The result follows from the Appendix~A of~\cite{AILSR15}.
\end{proof}

From the discussion above, we conclude that one can achieve the following
trade-off between space $n^{\rho_s + o(1)}$ and query time $n^{\rho_q
  + o(1)}$:
\begin{equation}
\label{eqn:treeTradeoff}
1 + \alpha^2 \rho_s - \rho_q - 2 \alpha \sqrt{\rho_s - \rho_q} = 0.
\end{equation}

We now show that this is equivalent to the tradeoff of
\cite{Laarhoven2015}, i.e.,~\eqref{eqn:LaaTradeoff}, where
$\rho_s=1+\rho_u$. Indeed, squaring~\eqref{eqn:treeTradeoff} and
replacing $\rho_s=1+\rho_u$, we
get:
\begin{equation}
\left((1 + \alpha^2)+ \alpha^2\rho_u - \rho_q\right)^2 = 4 \alpha^2(1+\rho_u - \rho_q),
\end{equation}
or
\begin{equation}
(1 + \alpha^2)^2+ \alpha^4\rho_u^2 +
  \rho_q^2+2\cdot((1+\alpha^2)\alpha^2\rho_u-(1+\alpha^2)\rho_q-\alpha^2\rho_u\rho_q)
  = 4 \alpha^2+4\alpha^2\rho_u - 4\alpha^2\rho_q.
\end{equation}
Simplifying the equation, we get
\begin{equation}
(1 - \alpha^2)^2+ \alpha^4\rho_u^2 +
  \rho_q^2+2\cdot((\alpha^2-1)\alpha^2\rho_u-(1-\alpha^2)\rho_q-\alpha^2\rho_u\rho_q)
  = 0.
\end{equation}

Remember that we have $\alpha=1-1/c^2$, and hence
$\alpha^2=\tfrac{(c^2-1)^2}{c^4}$ and
$1-\alpha^2=\tfrac{2c^2-1}{c^4}$. We further obtain:
\begin{equation}
\tfrac{(2c^2-1)^2}{c^8}+ \tfrac{(c^2-1)^4}{c^8}\rho_u^2 +
  \rho_q^2-2\tfrac{(2c^2-1)(c^2-1)^2}{c^8}\rho_u
-2\tfrac{2c^2-1}{c^4}\rho_q-2\tfrac{(c^2-1)^2}{c^4}\rho_u\rho_q
  = 0,
\end{equation}
or, multiplying by $c^8$,
\begin{equation}
\label{eqn:treeTrFinal}
(2c^2-1)^2+ (c^2-1)^4\rho_u^2 +
  c^8\rho_q^2-2(2c^2-1)(c^2-1)^2\rho_u
-2(2c^2-1)c^4\rho_q-2(c^2-1)^2c^4\rho_u\rho_q
  = 0.
\end{equation}

In a similar fashion, squaring~\eqref{eqn:LaaTradeoff}, we obtain:
\begin{equation}
c^4\rho_q+(c^2-1)^2\rho_u+2c^2(c^2-1)\sqrt{\rho_q\rho_u}=2c^2-1,
\end{equation}
or equivalently,
\begin{equation}
2c^2(c^2-1)\sqrt{\rho_q\rho_u}=2c^2-1-c^4\rho_q-(c^2-1)^2\rho_u.
\end{equation}
Squaring again, we obtain
\begin{multline}
4c^4(c^2-1)^2\rho_q\rho_u=(2c^2-1)^2+c^8\rho_q^2+(c^2-1)^2\rho_u^2 \\ +2\cdot\left(c^4(c^2-1)^2\rho_q\rho_u-(2c^2-1)c^4\rho_q-(2c^2-1)(c^2-1)^2\rho_u\right),
\end{multline}
or, simplifying,
\begin{equation}
(2c^2-1)^2+c^8\rho_q^2+(c^2-1)^2\rho_u^2-2c^4(c^2-1)^2\rho_q\rho_u-2(2c^2-1)c^4\rho_q-2(2c^2-1)(c^2-1)^2\rho_u=0
\end{equation}

We now observe that we obtain the same equation as~\eqref{eqn:treeTrFinal} and hence we are done proving that~\eqref{eqn:treeTradeoff} is equivalent to~\eqref{eqn:LaaTradeoff}.


\section{Upper Bound: General case}
\label{apx:upper_general}

We show how to extend the result of \cite{Laarhoven2015} (and
Appendix \ref{spherical_sec}) to the general case using the
techniques of~\cite{AR-optimal}. In particular, we show how to reduce
a~worst-case instance to several instances that are
\emph{random-like}. Overall the algorithm from below gives a data structure that solves the $(c,r)$-ANN problem in the
$d$-dimensional Euclidean space, using space $O(n^{1+\rho_u+o(1)}+dn)$, and
query time $O(dn^{\rho_q+o(1)})$ for any $\rho_u,\rho_q>0$ that satisfy:
\begin{equation}
c^2\sqrt{\rho_q}+(c^2-1)\sqrt{\rho_u}=\sqrt{2c^2-1}.
\end{equation}

As in~\cite{AR-optimal}, our data structure is a decision
tree. However, there are several notable differences from
\cite{AR-optimal}: 
\begin{itemize}
  \item The whole data structure is a \emph{single} decision tree, while in~\cite{AR-optimal} we consider a \emph{collection} of $n^{\Theta(1)}$ trees. 
  \item Instead of Spherical LSH used in~\cite{AR-optimal}, we use the partitioning procedure from Section~\ref{spherical_sec}.
  \item In~\cite{AR-optimal}, one proceeds with partitioning a dataset until all parts contain less than $n^{o(1)}$ points. We change the stopping criterion slightly to ensure the number of ``non-cluster'' nodes\footnote{think $K = O(\sqrt{\log n})$ as in Section~\ref{spherical_sec}.}  on any root-leaf branch is the same.
  \item Unlike~\cite{AR-optimal}, we do not use a ``three-point property'' of a random space partition in the analysis. This is related to the fact
  that the probability success of a single tree is constant, unlike~\cite{AR-optimal}, where it is polynomially small.
  \item In~\cite{AR-optimal} we reduce the general case to the ``bounded ball'' case using LSH from~\cite{DIIM}. Now we cannot quite do this, since we
  are aiming at getting a full time-space trade-off. Instead, we use a standard trick of imposing a randomly shifted grid, which reduces an arbitrary dataset to a dataset of diameter
  $\widetilde{O}(\sqrt{\log n})$~\cite{IM}. Then, we invoke an upper bound from~\cite{Laarhoven2015} together with a reduction from~\cite{Valiant12}, which for this case is enough to proceed. 
\end{itemize}

\subsection{Overview}
\label{sec_overview}

We start with a high-level overview. Consider a dataset $P_0$ of $n$ points. We can assume that
$r = 1$ by rescaling. 
We may also assume that the
dataset lies in the Euclidean space of dimension $d = \Theta(\log n
\cdot \log \log n)$: one can always reduce the dimension to $d$ by
applying Johnson-Lindenstrauss lemma~\cite{JL, DG03}
while incurring distortion at most $1 + 1 / (\log \log n)^{\Omega(1)}$
with high probability.

For simplicity, suppose that the entire dataset $P_0$ and a query lie
on a sphere $\partial B(0, R)$ of radius $R = O_c(1)$.  If $R \leq c /
\sqrt{2}$, we are done: this case corresponds to the ``random
instance'' of points and we can apply the data structure from
Section~\ref{spherical_sec}.

Now suppose that $R > c / \sqrt{2}$. We split $P_0$ into a number of
disjoint components: $l$ {\em dense} components, termed $C_1$, $C_2$,
\ldots, $C_l$, and one {\em pseudo-random} component, termed
$\widetilde{P}$.  The properties of these components are as follows.
For each dense component $C_i$ we require that $|C_i| \geq \tau n$ and
that $C_i$ can be covered by a spherical cap of radius $(\sqrt{2} -
\eps) R$ (see Fig.~\ref{cap_covering_fig}). Here $\tau, \eps > 0$ are small quantities to be chosen
later. The pseudo-random component $\widetilde{P}$ 
contains no more dense components inside.

\begin{figure}
    \begin{center}
        \includegraphics[page=2,scale=0.91]{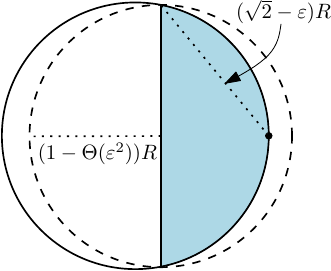}
    \end{center}
    \caption{Covering a spherical cap of radius $(\sqrt{2} - \eps) R$}
    \label{cap_covering_fig}
\end{figure}

We proceed separately for each $C_i$ and $\widetilde{P}$.
We enclose every dense component $C_i$ in slightly smaller ball $E_i$ of
radius $(1 - \Theta(\eps^2)) R$ (see Figure~\ref{cap_covering_fig}).
For simplicity, let us first ignore the fact that $C_i$ does not
necessarily lie on the boundary $\partial E_i$.  Once we enclose each dense cluster in a smaller ball, we recurse on each resulting spherical instance of radius $(1 -
\Theta(\eps^2)) R$.  We treat the pseudo-random part $\widetilde{P}$
as described in Section~\ref{spherical_sec}.  we sample $T$ Gaussian vectors
$z_1, z_2, \ldots, z_T \sim N(0, 1)^d$, where $T$ is a parameter to be chosen later
(for each pseudo-random remainder separately), and form $T$ subsets of $\widetilde{P}$ as follows:
$$
\widetilde{P}_i = \{p \in \widetilde{P} \mid \langle z_i , p \rangle \geq \eta R\},
$$ where $\eta > 0$ is a parameter to be chosen later (for each
pseudo-random remainder separately). Then we recurse on each
$\widetilde{P}_i$.  Note that after we recurse, there may appear new
dense clusters in some sets $\widetilde{P}_i$ (e.g., since it may
become easier to satisfy the minimum size constraint). 

During the query procedure, we recursively
query \emph{each} $C_i$ with the query point $q$. For the pseudo-random component
$\widetilde{P}$,  we identify all $i$'s such that $\langle z_i, q
\rangle \geq \eta' R$, and query all corresponding children
recursively. Here $\eta' > 0$ is a parameter to be chosen later (for
each pseudo-random remainder separately).

To analyze our algorithm, we show that we make progress in two ways.
First, for dense clusters we reduce the radius of a sphere by a factor
of $(1 - \Theta(\eps^2))$.  Hence, in $O_c(1 / \eps^2)$ iterations we
must arrive to the case of $R \leq c /\sqrt{2}$, which is easy (as
argued above).  Second, for the pseudo-random component $\widetilde P$,
we argue that most points lie at a distance $\ge(\sqrt{2} - \eps) R$
from each other. In particular, the ratio of $R$ to a typical
inter-point distance is $\approx 1/\sqrt{2}$, exactly like in a random case. This is the reason we call $\widetilde{P}$ pseudo-random.
This setting is where the data structure from Section~\ref{spherical_sec} performs well.

We now address the issue deferred in the above high-level
description: namely, that a dense component $C_i$ does not generally
lie on $\partial E_i$, but rather can occupy the interior of $E_i$. In this case, we partitioning $E_i$ into very thin annuli of carefully
chosen width $\delta$ and treat each annulus as a sphere. This
discretization of a ball adds to the complexity of the analysis,
but is not fundamental from the conceptual point of
view.

\subsection{Formal description}

We are now ready to describe the data structure formally. It depends
on the (small positive) parameters $\tau$, $\eps$ and $\delta$, as well as an integer parameter $K \sim \sqrt{\log n}$.
We also need to choose parameters $T$, $\eta > 0$, $\eta' > 0$ for each pseudo-random remainder separately.

\begin{figure}
    \begin{center}
        \includegraphics{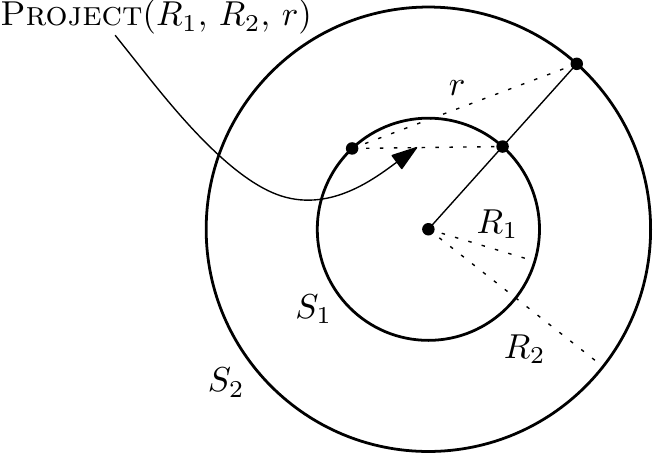}
    \end{center}
    \caption{The definition of \textsc{Project}}
    \label{project_fig}
\end{figure}

\paragraph{Preprocessing.}
Our preprocessing algorithm consists of the following functions:
\begin{itemize}
    \item \textsc{ProcessSphere}($P$, $r_1$, $r_2$, $o$, $R$, $k$) builds the data structure for a dataset
    $P$ that lies on a sphere $\partial B(o, R)$,
    assuming we need to solve ANN with distance thresholds $r_1$ and $r_2$. Moreover, we are guaranteed that queries
    will lie on $\partial B(o, R)$. The parameter $k$ is a counter which, in some sense, measures how far are we from being done.
    \item \textsc{ProcessBall}($P$, $r_1$, $r_2$, $o$, $R$, $k$) builds the data structure for a dataset 
    $P$ that lies inside the ball $B(o, R)$, assuming we need to solve ANN with distance thresholds $r_1$ and $r_2$.
    Unlike \textsc{ProcessSphere}, here queries can be arbitrary. The parameter $k$ has the same meaning as above.
    \item \textsc{Process}($P$) builds the data structure for a
      dataset $P$ to solve the general $(c,1)$-ANN;
    \item \textsc{Project}($R_1$, $R_2$, $r$) is an auxiliary
      function computing the following projection. Suppose we have two spheres $S_1$ and $S_2$
    with a common center and radii $R_1$ and $R_2$. Suppose there are points $p_1 \in S_1$ and $p_2 \in S_2$
    with $\|p_1 - p_2\| = r$. \textsc{Project}($R_1$,~$R_2$,~$r$) returns the distance between $p_1$ and the
    point $\widetilde{p_2}$ that lies on $S_1$ and is the closest to $p_2$ (see Figure~\ref{project_fig}).
\end{itemize}

We now elaborate on algorithms in each of the above functions.

\paragraph{\textsc{ProcessSphere}.}  Function \textsc{ProcessSphere}
follows the exposition from Section~\ref{sec_overview}. We consider
three base cases. First, if $k = K$, then we stop and store the whole
$P$.  Second, if $r_2 \geq 2R$, then the goal can be achieved
trivially, since any point from $P$ works as an answer for any valid
query. Third, if an algorithm from Section~\ref{spherical_sec} would
give a desired point on the time-space trade-off (in particular, if
$r_2 \geq \sqrt{2} R$), then we just choose $\eta, \eta' > 0$ and $T$
appropriately (in particular, we set $\eta > 0$ such that for any $u,
v$ with $\|u - v\| = r_2$ one has $\mathrm{Pr}_{z \sim N(0,
  1)^d}\left[\langle z, u \rangle \geq \eta R \mbox{ and } \langle z,
  v \rangle \geq \eta R\right] = n^{-1 / K} = 2^{-\sqrt{\log n}}$) and
make a single step . 

Otherwise, we find dense clusters, i.e., non-trivially smaller balls,
of radius $(\sqrt{2} - \eps)R$, with centers on $\partial B(o, R)$
that contain many data points (at least $\tau |P|$).  These balls can
be enclosed into balls (with unconstrained center) of radius
$\widetilde{R} \leq (1 - \Omega(\eps^2)) R$. For these balls we invoke
\textsc{ProcessBall} with the same $k$.  Then, for the remaining
points we perform a single step of the algorithm from
Section~\ref{spherical_sec} with appropriate $\eta, \eta' > 0$ and $T$
(in particular, we set $\eta > 0$ as above for the distance $\sqrt{2}
R$), and recurse on each part with $k$ increased by $1$.

\paragraph{\textsc{ProcessBall}.}
First, we consider the following simple base case. If $r_1 + 2R \leq r_2$,
then any point from $B(o, R)$ could serve as a valid answer to any query.

In general, we reduce to the spherical case via a discretization of
the ball $B(o, R)$. First, we round all the distances to $o$ up to a
multiple of $\delta$, which can change distance between any pair of
points by at most $2 \delta$ (by the triangle inequality).  Then, for
every possible distance $\delta i$ from $o$ to a~data point and every
possible distance $\delta j$ from $o$ to a~query (for admissible
integers $i,j$), we build a~separate data structure via
\textsc{ProcessSphere} (we also need to check that $|\delta (i - j)|
\leq r_1 + 2 \delta$ to ensure that the corresponding pair $(i, j)$
does not yield a trivial instance).  We compute the new
distance thresholds $\widetilde{r}_1$ and $\widetilde{r}_2$ for this
data structure as follows.  After rounding, the new thresholds for the
ball instance should be $r_1 + 2 \delta$ and $r_2 - 2 \delta$, since
distances can change by at most $2 \delta$. To compute the final thresholds
(after projecting the query to the sphere of radius $\delta i$), we just
invoke \textsc{Project} (see the definition above).

\paragraph{\textsc{Process}.}  \textsc{Process} reduces the general
case to the ball case. We proceed similarly to \textsc{ProcessSphere},
with two modifications. First, we apply a randomized partition using
cubes with side $O_c(\sqrt{d}) = \widetilde{O}_c(\sqrt{\log n})$, and
solve each part separately.  Second, we seek to find dense clusters of
radius $O_c(1)$. After there are no such clusters, we apply the
reduction to unit-norm case from~\cite[Algorithm 25]{Valiant12}, and
then (a single iteration of) the algorithm from
Section~\ref{spherical_sec}.

\paragraph{\textsc{Project}.}
This is implemented by a formula as in \cite{AR-optimal} (see
Figure~\ref{project_fig}).

Overall, the preprocessing creates a decision tree, where the nodes
correspond to procedures {\sc ProcessSphere}, {\sc ProcessBall}, {\sc
  Process}. We refer to the tree nodes correspondingly, using the
labels in the description of the query algorithm from below.

\paragraph{Query algorithm.}
Consider a query point $q \in \Rbb^d$. We run the query on the
decision tree, starting with the root, and applying the following
algorithms depending on the label of the nodes:
\begin{itemize}
\item In \textsc{Process} we first recursively query the data
  structures corresponding to the clusters.
Second, we locate $q$ in the spherical caps, and query the data structure we built for the corresponding subsets of $P$.
\item In \textsc{ProcessBall}, we first consider the base case, where we just return the stored point if it is
close enough. In general,
we check if $\|q - o\| \leq R + r_1$. If not, we can return.
Otherwise, we round $q$ so that the distance from $o$ to $q$ is a multiple of $\delta$.
Next, we enumerate the distances from $o$ to the potential near neighbor we are looking for,
and query the corresponding {\sc ProcessSphere} children after
projecting $q$ on the sphere with a tentative near neighbor (using, naturally, \textsc{Project}).
\item In \textsc{ProcessSphere}, we proceed exactly the same way as
  \textsc{Process} modulo the base cases.
\item In all the cases we try all the points if we store them
  explicitly (which happens when $k = K$).
\end{itemize}

\subsection{How to set parameters}

Here we briefly state how one sets the parameters of the data structure.

Recall that the dimension is $d =
\Theta( \log n \cdot \log \log n)$.  We set $\eps, \delta, \tau$ as
follows:
\begin{itemize}
    \item $\eps = \frac{1}{\log \log \log n}$;
    \item $\delta = \exp\bigl(-(\log \log \log n)^C\bigr)$;
    \item $\tau = \exp\bigl(-\log^{2/3} n\bigr)$,
\end{itemize}
where $C$ is a sufficiently large positive constant.

Now we need to specify how to set $\eta, \eta' > 0$ and $T$ for each pseudo-random remainder. The idea is to set $\eta$, $\eta'$ and $T$ such that
$$
\mathrm{Pr}_{z \sim N(0, 1)^d}\left[\langle z, u \rangle\right] = n^{-1/K} = 2^{-\sqrt{\log n}}
$$
while, at the same time, for every $u$ and $v$ at distance at most $r_1$
$$
T \sim \frac{100}{\mathrm{Pr}_{z \sim N(0, 1)^d}\left[\langle z, u \rangle \geq \eta, \langle z, v \rangle \geq \eta'\right]}.
$$
Finally, we choose $T$ such that $T^K \sim n^{\rho_s + o(1)}$ where $\rho_s \geq 1$ is a parameter that governs the memory consumption.

This gives us a unique value of $\eta' > 0$, which governs the query time.

A crucial relation between parameters is that $\tau$ should be much smaller than $n^{-1/K} = 2^{-\sqrt{\log n}}$. This implies that the ``large distance'' is effectively
equal to $\sqrt{2} R$, at least for the sake of a single step of the random partition.

\end{document}